\documentclass[global]{svjour}
\usepackage{latexsym}
\usepackage{amsmath}
\usepackage{amsfonts,amssymb}
\usepackage{graphics}
\def\bra#1{\langle #1 |}
\def\ket#1{| #1 \rangle}
\newcommand{\tr}{\mathop{\text{Tr}}\nolimits}
\renewcommand{\Re}{\mathop{\mathrm{Re}}\nolimits}

\journalname{}
\begin{document}
\title{Multipartite entanglement in qubit systems}
\author{Paolo Facchi\inst{1,2} 
}                     
%
%
\institute{Dipartimento di Matematica, Universit\`a di Bari,
        I-70125  Bari, Italy \and INFN, Sezione di Bari, I-70126 Bari, Italy}
%
\date{\today}
%
\maketitle
\begin{abstract}
We introduce a potential of multipartite entanglement for a system of $n$ qubits,
as the average over all balanced bipartitions of a bipartite entanglement measure, the purity.
We study in detail its expression and look for its minimizers, the maximally multipartite entangled states. They  have a bipartite entanglement that does not depend on the
bipartition and is maximal for all possible bipartitions.  We investigate their structure and consider several examples for  small $n$.
\end{abstract}

\section{Introduction}
\label{intro}

Entanglement is one of the most striking features of quantum phenomena
\cite{1}. It plays very important roles in quantum information processing such as quantum
computation \cite{2}, quantum teleportation \cite{3} (for discussions on experimental realizations see
\cite{4.1,4,4.2,4.3}), dense coding \cite{5} and quantum cryptographic schemes \cite{6.1,6,6.2}.
Nevertheless, the quantification of multipartite entanglement is no simple matter.

Entanglement is intimately related to the very mathematical structure of quantum mechanics and complex Hilbert spaces. In particular it is a straightforward consequence of linearity (superposition principle) in tensor product Hilbert spaces (composite quantum systems). 

Consider a quantum system composed of two parts (e.g. two particles):  part $A$, whose Hilbert space is $\mathcal{H}_A$, and part $B$, whose Hilbert space is $\mathcal{H}_{B}$. According to quantum mechanics, the composite system lives in the tensor product Hilbert space $\mathcal{H}=\mathcal{H}_A\otimes\mathcal{H}_{B}$. The most familiar example is that of  two spinless particles, whose Hilbert space is
$L^2(\mathbb{R}^3)\otimes L^2(\mathbb{R}^3)\cong L^2(\mathbb{R}^6).$
The linearity of $\mathcal{H}$ implies that the  states $\ket{\psi}$  of the composite system $\mathcal{H}$ are linear combinations of product states, namely,
$$\ket{\psi}=\sum_{ij} z_{ij} \ket{\varphi_i}\otimes \ket{\chi_j},$$
with $\ket{\varphi_i}\in\mathcal{H}_A$ and
 $\ket{\chi_j}\in\mathcal{H}_{B}$. This entails interference among  probability amplitudes of two-particle states, that is the analogous of  the simpler case of one-particle interference. For example, the probability amplitude of  having both particle $A$ in state $\ket{\varphi_1}$ and particle $B$ in state $\ket{\chi_1}$ interferes with the probability amplitude of having both particle $A$ in state $\ket{\varphi_2}$ and particle $B$ in state $\ket{\chi_2}$. As a consequence there exist correlations of quantum nature --entanglement-- between quantum subsystems. These correlations are stronger than the classical ones, in the sense that they violate a class of inequalities, named after Bell, that must be satisfied by all classical correlations \cite{1}.

 The most striking violation of Bell's inequalities  is given by a particular class of states: maximally entangled states.  The simplest example is that of two spin-$1/2$ systems (or qubits), whose Hilbert space is $\mathbb{C}^2\otimes\mathbb{C}^2$, in the  singlet state
 $$\ket{\Phi}=(\ket{0}\otimes \ket{1} - \ket{1}\otimes \ket{0})/\sqrt{2}, $$
 where $\{\ket{0},\ket{1}\}$ is the natural basis of $\mathbb{C}^2$, representing spin up or down in a given direction.
 The expectation value of any local observable $O$ of the first spin is given by
 $$\bra{\Phi} O\otimes 1\ket{\Phi}=  \frac{1}{2} \bra{0}O \ket{0} + \frac{1}{2} \bra{1}O \ket{1}= \tr (\rho_A O),$$
 and thus is an incoherent average corresponding to a completely mixed reduced density matrix of the first spin $\rho_A= 1/2$. Analogously for the second spin.
Therefore, spin measurements  in a given, arbitrary, direction over an ensemble of pairs prepared in a singlet state will result in a completely random sequence of 0 and 1. On the other hand, the results of joined local measurements exhibit strong correlations, due to the fact that the total spin is 0: the two spins are always found pointing in opposite directions. The two random sequences are exactly complementary.

Maximally entangled states are characterized by the property, just shown for the two-qubit singlet state,
 that to a perfect knowledge of the state of the composite system corresponds a complete ignorance of the states of its two parts.
More precisely, although the composite system is in a well determined pure state, its two parts are in completely mixed states. See Corollary \ref{cor:mbe}.
 Therefore, all information is totally shared by the two parts.
Note that this situation is strongly at variance with the classical case, in which a complete knowledge of the total system is equivalent to a complete knowledge of both its parts. In quantum mechanics this is only a necessary condition.

In general, the degree of bipartite entanglement of  a quantum system can range from a maximum, when its two parts are in  completely mixed states, to a minimum, when its two parts are in pure states, and thus a complete knowledge of both subsystems implies a complete knowledge of the entire system, as in the classical case.
This is the case of separable, or unentangled, states $\ket{\psi}=\ket{\varphi}\otimes\ket{\chi}$, that have no correlations between the two parts and, thus, no shared information.

The degree of bipartite entanglement \cite{woot} of a composed quantum system
 can be consistently quantified, as in Definition \ref{def:puritydef}, in terms of the purity of the reduced density matrix of one of the two subsystems (purity can be proven to be the same for both, see Lemma \ref{lem:reduced}). A lower value of purity will correspond to a larger value of entanglement.

On the other hand, there is no unique way of quantifying
multipartite entanglement \cite{entanglement1}, that is entanglement among $n$ given parties of a given quantum system.
 Different definitions often do not
agree with each other, because they adopt different strategies,
focus on different aspects and capture different features of this
quantum phenomenon \cite{druss,multipart,MMSZ,mw,multipart1}.  There is a profound
reason behind this manifest disadvantage: the number of
 real numbers, i.e.\ the invariants under local unitary transformations \cite{Albeverio1,Albeverio3,Albeverio2},
   needed to quantify multipartite entanglement
grows exponentially with the size of the system, so that the definition of appropriate entanglement measures, able to summarize the most salient global features of
entanglement, can be very difficult.

A natural generalization of the bipartite case is to quantify the entanglement among $n$ parties by considering the average purity over  subsystems \cite{MMES,entrandom}. In this paper we will follow this strategy and, in particular, we will consider systems of $n$ qubits.

After introducing  notation and discussing some results about bipartite entanglement in Section \ref{sec:bipartite}, we move to multipartite entanglement and study the properties of the potential of multipartite entanglement (i.e.\ the average purity over balanced bipartitions), Definition \ref{def:MMES}, and of its minimizers, i.e. quantum states with the maximal degree of multipartite entanglement. In the ideal situation, the bipartite entanglement of such states is not only maximal, but also does not depend on the way one decides to bipartite the total system into two subsystems. See Definition \ref{def:perfect}.
Our approach is  based on the action of the permutation group on the Fourier coefficients of the quantum state and thus is of combinatoric nature.

The potential of multipartite entanglement, Eq. (\ref{eq:pimeDelta}) of Theorem \ref{th:piDelta}, is a quartic Hamiltonian
\begin{equation*}
\pi_{\mathrm{ME}}(z)
=  \sum_{k,k',l,l'} \Delta(k, k'; l, l';[n/2])\, z_{k}\, z_{k'}\,
\bar{z}_{l}\, \bar{z}_{l'}\, ,
\end{equation*}
where  
$z=(z_k)$, with $k=(k_1,\cdots,k_n) \in \{0,1\}^n$ is the vector of the 
Fourier coefficients of the state of a system of $n$ qubit in the computational basis
\begin{equation*}
|\psi\rangle=\sum_k z_k |k_1\rangle \otimes |k_2\rangle \otimes \cdots |k_n\rangle.
\end{equation*}
One of the aims of this paper is to investigate the complex structure of the   long-range coupling function $\Delta$, that appears in the above expression. This is accomplished in Theorems \ref{th:Delta},  \ref{th:pme1} and  \ref{th:pime4}, and in Corollary \ref{th:pme2}. In particular, a measure of the complexity of the potential of multipartite entanglement is given by the number of its nontrivial interfering terms,
that scales  like $2^{n-3} 3^n$ (see Theorem \ref{th:pme3e} and the following remark).

Sections \ref{sec:perfect MMES} and \ref{sec:uniform MMES} will then be devoted to investigate
 maximally multipartite entangled states (MMES) \cite{MMES}, i.e.\  the minimizers of the potential of multipartite entanglement. In particular, the structure of  perfect MMES, i.e. minimizers that are maximally entangled with respect to any bipartition, is analyzed in Section \ref{sec:perfect MMES}: by making use of a probabilistic approach, Theorem \ref{th:populationMMES} gives a complete characterization of their population probability vectors $(|z_k|^2)$, while Theorem \ref{th:phasesperfectMMES} exhibits the equations that must be satisfied by their phases $(\zeta_k)=(z_k/|z_k|)$. The number of  equations quickly overcomes the number of variables, since their ratio scales  as $2^{n+1}/\sqrt{n}$ with the number  $n$ of qubits. See Theorem 
\ref{th:eqs vs vars} and following remark. Therefore, for large systems it becomes more and more difficult to have a perfect MMES solution, unless symmetries subtly conspire to reduce the number of independent equations. In fact, the existence of perfect MMES for $n\leq 6$,  $n\neq 4$,  will be proven by explicit construction in Section \ref{sec:uniform MMES}, while it is known  \cite{scott,rains1,rains3,rains2}  that for $n\geq 8$ they do not exist. The case $n=7$ remains open, although there is numerical evidence that no perfect MMES exist \cite{MMES}.
In conclusion, apart from some special small values of $n$, not all bipartitions can have minimal purity (maximal entanglement) and the requirement that a given bipartition be in a maximally entangled state collides with the same requirement for a different bipartition. Thus, the bipartitions of a general MMES  are in a frustrated configuration, and this makes the whole subject richer and very interesting.  

Since, according to the structure theorem \ref{th:populationMMES}, a perfect MMES can have a uniform population probability vector $(|z_k|^2)=(1/N,\cdots,1/N)$  with $N=2^n$, in Section \ref{sec:uniform MMES} we focus on this class of uniform states, and restrict our quest to it.
 We will explicitly construct perfect MMES with uniform population, and will easily characterize the fully factorized states, i.e.\ the maximizers of the potential of multipartite entanglement that have uniform probability vectors.  By pushing even further our simplifying assumptions, we will explicitly show that, at least for $n\leq 6$, the potential admits minimizers and maximizers in the very restricted class of uniform states with real phases, $(z_k)=(\zeta_k/\sqrt{N})$,  with  $\zeta_k \in \{+1,-1\}$. This states can be naturally mapped onto the set of binary sequences  of length  $N=2^n$, and the potential of multipartite entanglement becomes a quartic Hamiltonian on binary sequences (or classical spins). It is then quite natural to investigate whether there is any relation between the minimizing sequences of $\pi_{\mathrm{ME}}$ and the low correlation sequences that minimize similar long-range Hamiltonians studied in \cite{MPR1,MPR2,Borsari1,Borsari2}, which quantify all possible correlations in a binary string.
However, we will leave this problem for a future publication.

A final remark is in order. The study of the minimizers of $\pi_{\mathrm{ME}}(z)$ can be embedded in a statistical mechanical framework \cite{statmec}. Let us consider the partition function of a system with Hamiltonian $\pi_{\mathrm{ME}}(z)$ at a fictitious temperature $\beta^{-1}$,
\begin{equation*}
Z_N=\int  \exp\big(-\beta\, \pi_{\mathrm{ME}}(z)\big)\; \mathrm{d}\mu(z),
\end{equation*}
where $\mu$ is the uniform measure (of typical states) on the hypersphere $\{z\in \mathbb{C}^N \, | \, \sum_k |z_k|^2=1\}$ induced by the Haar measure on $U(N)$.
The value of the free energy  $F_N(\beta)=-\beta^{-1} \ln Z_N$, will become that of the minimum of the Hamiltonian $\pi_{\mathrm{ME}}(z)$ when the temperature tends to zero, that is $\beta\to\infty$, and only those
configurations that minimize the Hamiltonian survive, namely the maximally multipartite entangled states.
 In general $\beta$, as a Lagrange multiplier, fixes the average value of  entanglement, larger values of $\beta$ corresponding to a higher multipartite entanglement. In particular,
for $\beta\to0$ one is looking at the typical states.  Remarkably, there is a physically appealing interpretation
even for negative temperatures: for $\beta\to -\infty$, those
configurations are selected that maximize the Hamiltonian, that is
fully factorized states.

This approach has proven to be very useful in the (much simpler) case of bipartite entanglement, when the potential of multipartite entanglement reduces to the purity $\pi_A$ of one of the two component subsystems, and in the thermodynamical  limit $N\to\infty$,  
the existence of two phase transitions, characterized by different spectra of the reduced density matrices, has been shown \cite{matrix}. 

In order to investigate the statistical mechanics of the richer and more complex case of multipartite entanglement, and possibly to unveil its phase transitions, it is necessary to study in detail the structure of the potential of multipartite entanglement $\pi_{\mathrm{ME}}$ and of the highly entangled states that give rise to its low energy landscape, the MMES.  
This paper is completely devoted to such a study.

\section{Bipartite entanglement}
\label{sec:bipartite}

In this section we will set up the notation and we will prove some results about bipartite entanglement that will be used in the following. We will show how the entanglement of a bipartite system in a pure state is related to the non-vanishing eigenvalues of the reduced density matrix of one of its parts. In particular, in Corollary \ref{cor:mbe} we will show that in an unentangled, separable, state of a composed system the reduced density matrices of its two parts are pure, i.e.\ are 1-dimensional projections and thus have only one non-vanishing eigenvalue, that equals 1. On the other hand, a bipartite system is in a maximally entangled state if and only if the reduced density matrix of its smaller part is completely mixed, i.e.\ is proportional to the identity operator and all its eigenvalues are equal and different from 0. 

Therefore, as a measure of bipartite entanglement one can use the purity, i.e.\ the sum of the squared eigenvalues, of the reduced density matrix of the smaller party.  We will do this in Definition \ref{def:puritydef}. 
One can show that purity ranges in a compact interval, its minimum corresponding to maximally bipartite entangled states and its maximum to the bipartite separable ones. This simple result, which is the content of Lemma \ref{lm:purityconstraint}, together with the explicit expression of the purity as a function of the Fourier coefficients  of the state, as given in Theorem \ref{th:rhoA piA} and its corollary, will play a crucial role in the following.

Let us start with some basic definitions.

\begin{definition}[Qubit]
A \emph{qubit} (or \emph{spin}) is a quantum system with a two-dimensional Hilbert space $\mathfrak{h} \cong\mathbb{C}^2$.
The \emph{computational}  \emph{basis} $\{\ket{0}, \ket{1}\}$ is a privileged orthonormal basis.
\end{definition}
\begin{definition}[System of qubits]
\label{def:system}
A system $S=\{1,2,\dots, n\}$ of $n$ qubits is a quantum system with a $2^n$ dimensional Hilbert space
$\mathcal{H}_S= \bigotimes_{i\in S} \mathfrak{h}_i$, with $\mathfrak{h}_i\cong\mathbb{C}^2$.  Its
\emph{pure states} are the normalized vectors $\ket{\psi}\in\mathcal{H}_S$ with $\bra{\psi}\psi\rangle=1$, and can be expressed in the computational bases as
\begin{equation}
|\psi\rangle = \sum_{k\in X^n} z_k |k\rangle , \quad
z_k\in \mathbb{C}, \quad \sum_{k\in X^n} {|z_k|}^2 =1,
\label{eq:genrandomx}
\end{equation}
where $k=(k_i)_{i\in S}=(k_1,k_2,\dots,k_n)$, with $k_i\in X=\{0,1\}$, and
\begin{equation}
\label{eq:ki}
\ket{k}=\ket{k}_S=\bigotimes_{i\in S} \ket{k_i}_i, \qquad
\ket{k_i}_i \in \mathfrak{h}_i.
\end{equation}
\end{definition}
\begin{definition}[Bipartition]
A \emph{bipartition} of the system $S$ is a pair $(A,\bar{A})$, with $1\leq n_A \leq n_{\bar{A}}$,
where $A \subset S$, $\bar{A}=S\backslash A$ (i.e.\ $S=A+\bar{A}$) and $n_A=|A|$, the cardinality of $A$.
The bipartition is said to be \emph{balanced} if $A$ is maximal, that is $n_A=\left[n/2\right]$ (and $n_{\bar{A}}=\left[(n+1)/2\right]$), with $[x]=\;$integer part of $x$.
\end{definition}
\begin{remark}
There is a one to one correspondence among bipartitions and nonempty subsets of $S$ of
dimension not exceeding $n/2$. Given a bipartition $(A,\bar{A})$, the total Hilbert space is accordingly partitioned
into $\mathcal{H}_S=\mathcal{H}_A\otimes\mathcal{H}_{\bar{A}}$,
 where $\mathcal{H}_A= \bigotimes_{i\in A} \mathfrak{h}_i$,  with $N_A = \dim\mathcal{H}_A=2^{n_A}$, is the Hilbert space of the ensemble $A$ of $n_A$ qubits.
\end{remark}
\begin{definition}[Entanglement]
A state $\ket{\psi}\in\mathcal{H}_S$ is said to be \emph{separable} with respect to the bipartition $(A,\bar{A})$ if it can be expressed as a tensor product $\ket{\psi}=\ket{\phi}_A \otimes \ket{\chi}_{\bar{A}}$ for some
$\ket{\phi}_A\in\mathcal{H}_A$ and $\ket{\chi}_{\bar{A}}\in\mathcal{H}_{\bar{A}}$.
A state that is not separable is called \emph{entangled}.
\end{definition}
The following lemma is a powerful tool in the study of entanglement.
\begin{lemma}[Schmidt decomposition]
\label{schmidt}
Given a bipartition $(A,\bar{A})$, every state $\ket{\psi} \in\mathcal{H}_S$ can be written in the form
\begin{equation}
\label{eq:Schmidt decomposition}
\ket{\psi}=\sum_{k\in Y } \sqrt{\lambda_k}  \ket{u_k}_A\otimes\ket{v_k}_{\bar{A}},
\end{equation}
with $\lambda_k > 0$, $\sum \lambda_k=1$, $Y \subset X^{n_A}$, and where
 $\{\ket{u_k}_A\}\subset\mathcal{H}_A$, $\{\ket{v_k}_{\bar{A}}\}\subset\mathcal{H}_ {\bar{A}} $ are orthonormal sets. The set of \emph{Schmidt coefficients} $\{\lambda_k\}$ is unique.
\end{lemma}
\begin{proof}
In the computational basis
$\ket{\psi}=\sum_{l\in X^{n_A}} \sum_{m\in X^{n_{\bar{A}}}}
 t_{lm}\ket{l}_A\otimes\ket{m}_{\bar{A}}.$
The matrix $t=(t_{lm})$ considered as an operator from $\mathbb{C}^{N_{\bar{A}}}$ to $\mathbb{C}^{N_A}$ admits a singular value decomposition
 $t=\sum_{k\in Y}  \sqrt{\lambda_k}  u^{(k)} v^{(k)*}$, for some $Y\subset X^{n_A}$, with $\{u^{(k)}\}\subset \mathbb{C}^{N_A}$ and $\{v^{(k)}\}\subset \mathbb{C}^{N_{\bar{A}}}$ orthonormal sets, and $\lambda_k>0$ \cite{lang}.  One gets $1=\bra{\psi}\psi\rangle=\tr{(t^* t)}=\sum\lambda_k$. The desired result immediately follows, with
$\ket{u_k}_A=\sum_l u^{(k)}_l \ket{l}_A$ and $\ket{v_k}_A=\sum_m \bar{v}^{(k)}_m \ket{m}_{\bar{A}}$.
\qed
\end{proof}
It follows immediately that
\begin{theorem}
\label{prop:separable}
A bipartite state $\ket{\psi}$ is separable with respect to the bipartition $(A,\bar{A})$ iff the set of Schmidt coefficients reduces to $\{1\}$.
\qed
\end{theorem}
\begin{remark}
\label{rem:dist}
In general, one wants also a measure that quantifies the entanglement of a bipartite state, i.e.\  how much the given state differs from a separable one. To this purpose, note that one can associate to the Schmidt coefficients of a given bipartition $\{\lambda_l, l\in Y\subset X^{n_A}\}$ a probability distribution $p$ over the finite space $X^{n_A}$, in the following way:
\begin{equation}
p(l)=\begin{cases}
\lambda_l & \text{for}\; l\in Y\\
0 & \text{otherwise}.
\end{cases}
\end{equation}
Therefore, it is natural to consider as a measure of bipartite entanglement the distance of the probability vector $p$ from the set SEP of the separable vectors,  concentrated at a point,
\begin{equation}
\label{eq:SEP}
\mathrm{SEP}=\{p(\cdot)=\delta_\ell(\cdot)\}_{\ell\in X^{n_A}},
\end{equation}
where $\delta_\ell(\cdot)=\delta_{\{\ell\}}(\cdot)$. Here $\delta_C$ is the characteristic function of set $C$,
\begin{equation}
\label{eq:charfun}
\delta_C(x) =\begin{cases} 1 & \text{if} \; x\in C\\
0 & \text{if}\; x\notin C.
\end{cases}
\end{equation}
We consider the distance derived from the $L^1$ norm,
\begin{equation}
d(p_1,p_2)=\frac{1}{2}\sum_{l\in X^{n_A}}|p_1(l)-p_2(l)| .
\end{equation}
It is easy to see that $0\leq d(p_1,p_2)\leq 1$ and that $d(p_1,p_2)=\sum_{l}[p_1(l)-p_2(l)]_+$, where $[\cdot]_+$ denotes the positive part.
Therefore, if $p$ is the probability vector associated to the Schmidt coefficients $\{\lambda_l\}$ of the state $\ket{\psi}$ in the bipartition $(A,\bar{A})$, one gets
\begin{equation}
\min_{q\in\mathrm{SEP}} d(p,q)=1-\max_l \lambda_l.
\end{equation}
\end{remark}
This motivates the following
\begin{definition}[Entanglement measure]
\label{def:dist}
A measure of the entanglement of state $\ket{\psi}$ with respect to the bipartition $(A,\bar{A})$ is given by
\begin{equation}
\mathcal{E}_A(\ket{\psi})=\frac{N_A}{N_A-1}\left(1-\max_l \lambda_l\right),
\end{equation}
where the maximum is taken over the set of the Schmidt coefficients $\{\lambda_l\}$ of the state in the given bipartition and $N_A=2^{n_A}$.
\end{definition}
By noting that, due to normalization, $N_A^{-1}\leq\max_k \lambda_k\leq 1$,
it follows that
\begin{theorem}
One gets $0\leq\mathcal{E}_A(\ket{\psi})\leq 1$. Moreover $\mathcal{E}_A(\ket{\psi})=0$ iff $\ket{\psi}$ is separable with respect to the bipartition $(A,\bar{A})$ . \qed
\end{theorem}
On the other hand, states that maximize the entanglement measure $\mathcal{E}_A$ are the main interest of this paper
\begin{definition}[Maximally bipartite entangled states]
\label{def:mbes}
A state $\ket{\psi}$ that satisfies $\mathcal{E}_A(\ket{\psi})= 1$ is called a
\emph{maximally bipartite entangled} state  with respect to the bipartition $(A,\bar{A})$ .
\end{definition}
\begin{theorem}[Local unitary invariance]
\label{prop:mbes}
A state $\ket{\psi}$ is maximally entangled with respect to the bipartition $(A,\bar{A})$ iff
\begin{equation}
\label{eq:mbes}
|\psi\rangle = N_A^{-1/2} \sum_{l\in X^{n_A}} U^A \ket{l}_A\otimes U^{\bar{A}} \ket{l}_{\bar{A}} ,
\end{equation}
where $U^A$ and $U^{\bar{A}}$ are (local) unitary operators in $\mathcal{H}_A$ and $\mathcal{H}_{\bar{A}}$, respectively.
\end{theorem}
\begin{proof}
A state $\ket{\psi}$ is maximally  entangled  iff  $\mathcal{E}_A(\ket{\psi})= 1$, i.e.,
$\max_k \lambda_k=1/N_A$. Thus its probability vector is completely mixed, $\lambda_k=1/N_A$ $\forall k\in X^{n_A}$. From Lemma~\ref{schmidt} one gets the thesis  where $U^A$ ($U^{\bar{A}}$) is the local unitary operator in $\mathcal{H}_A$ ($\mathcal{H}_{\bar{A}}$)  that transforms the computational basis into the Schmidt one, namely  $U^A\ket{l}_A=\ket{u_l}_A$ ($U^{\bar{A}}\ket{l}_{\bar{A}}=\ket{v_l}_{\bar{A}}$).
\qed
\end{proof}
\begin{remark}
Note that Eq. (\ref{eq:mbes}) implicitly assumes an arbitrary embedding of $X^{n_A}$ in $X^{\bar{A}}$ and thus, when $n_{\bar{A}}>n_A$, relies on an arbitrariness in the choice of the subset $\{\ket{l}_{\bar{A}}\}_{l\in X^{n_A}}$ of the computational basis of party $\mathcal{H}_{\bar{A}}$.
\end{remark}

\begin{remark}
Note that, while separable states (\ref{eq:SEP}) are associated to extremal probability vectors, concentrated at a point, maximally entangled bipartite states are associated to completely mixed probability vectors, uniform on $X^{n_A}$. By Theorem~\ref{prop:mbes}, the above property can be used as an equivalent definition of maximally entangled bipartite states. It has the advantage of being independent of the particular measure $\mathcal{E}_A$.
\end{remark}
An immediate consequence of Theorem \ref{prop:mbes} is the following
\begin{corollary}
A maximally bipartite entangled state has the following Fourier coefficients in the computational basis
\begin{equation}
\label{eq:maxbi}
z_k=N_A^{-1/2} \sum_{l\in X^{n_A}} U^A_{k_A,l} U^{\bar{A}}_{k_{\bar{A}},l} , \qquad
k\in X^n,
\end{equation}
where $N_A=2^{n_A}$ and $U^A_{l,l'}=\bra{l_A} U^A  \ket{l'_A}$ with $U^A$ the local unitary operator in $\mathcal{H}_A$ that transforms the computational basis into a Schmidt one. \qed
\end{corollary}
In fact, there is a link between the set of Schmidt coefficients and the reduced density matrices of subsystems $A$ and $\bar{A}$. Recall that
\begin{definition}[Reduced density matrix]
\label{def:reduced}
If $\rho$ is a density matrix on $\mathcal{H}_A\otimes\mathcal{H}_{\bar{A}}$, then the reduced density matrix $\rho_A$ is a density matrix on $\mathcal{H}_A$ defined by
\begin{equation}
\rho_A=\tr_{\bar{A}}\rho,
\end{equation}
where $\tr_{\bar{A}}$ is the partial trace over subsystem $\bar{A}$.
\end{definition}
\begin{remark}
The reduced density matrix represents the state of a subsystem $A$, since it determines the statistics of every (local) observables of $A$.
\end{remark}
\begin{lemma}[Reduced density matrix eigenvalues] \cite{arakilieb}
\label{lem:reduced}
Given a pure state in $\mathcal{H}_S$, the reduced density matrices $\rho_A$ and $\rho_{\bar{A}}$ of subsystems $A$ and $\bar{A}$ have the same eigenvalues and multiplicities, except possibly for the eigenvalue 0.
\end{lemma}
\begin{proof}
From Lemma~\ref{schmidt}, one gets
$\rho_A=\tr_{\bar{A}}\ket{\psi}\bra{\psi}=\sum \lambda_k \ket{u_k}\bra{u_k}$, and
$\rho_{\bar{A}}=\tr_{A}\ket{\psi}\bra{\psi}=\sum \lambda_k \ket{v_k}\bra{v_k}$.
\qed
\end{proof}
\begin{remark}
From the proof of Lemma~\ref{lem:reduced} one sees that the Schmidt coefficients of a bipartite state are the nonzero eigenvalues of the partial density matrices of the two parties (and the vectors of the Schmidt decomposition are the corresponding eigenvectors). Therefore, from Remark~\ref{rem:dist} and Definition~\ref{def:dist} we obtain
\end{remark}
\begin{corollary}
Given a state $\ket{\psi}\in\mathcal{H}_S$ and a bipartition $(A,\bar{A})$ one gets
\begin{eqnarray}
\mathcal{E}_A(\ket{\psi})&=&\frac{N_A}{N_A-1} \min\left\{\left\|\rho_A - \ket{\phi}\bra{\phi}\right\|_1 :
\ket{\phi}\in\mathcal{H}_A,\bra{\phi}\phi\rangle=1\right\}
\nonumber\\
&=&\frac{N_A}{N_A-1}\left(1-\| \rho_A \|\right),
\end{eqnarray}
where $\| \cdot\|_1=\tr |\cdot|$ is  the trace norm and $\|\cdot\|$ is the operator norm. \qed
\end{corollary}
Moreover,
\begin{corollary}
\label{cor:mbe}
Given a bipartition $(A,\bar{A})$, a state $\ket{\psi}\in\mathcal{H}_S$ is separable iff $\rho_A=\ket{\phi}\bra{\phi}$ for some normalized $\ket{\phi}\in\mathcal{H}_A$ and is maximally entangled iff
$\rho_A=1/N_A$. \qed
\end{corollary}

As an alternative measure of the bipartite entanglement between the two
subsets, which is more suitable to analytical treatment, we consider the linear entropy of subsystem $A$.
\begin{definition}[Linear entropy and purity]
\label{def:puritydef}
A measure of the entanglement of state $\ket{\psi}$ with respect to the bipartition $(A,\bar{A})$ is given by
\begin{equation}
\mathcal{L}_A(\ket{\psi})=\frac{N_A}{N_A-1}\left(1-\pi_A(\ket{\psi}) \right),
\label{eq:linent}
\end{equation}
where $N_A=2^{n_A}$, and
\begin{equation}
\pi_A(\ket{\psi})=\tr_A\rho_{A}^2, \quad \rho_{A}=\tr_{\bar{A}}|\psi\rangle\langle\psi|
\end{equation}
is the \emph{purity} of subsystem $A$.
\end{definition}
By noting that $\pi_A(\ket{\psi})=\sum_l \lambda_l^2$, where $\{\lambda_l\}$ is the set of the Schmidt coefficients of the state in the given bipartition it follows that
\begin{lemma}[Purity bounds]
\label{lm:purityconstraint}
Given a state $\ket{\psi}\in\mathcal{H}_S$ and a  bipartition $(A,\bar{A})$, one has $\pi_{A}(\ket{\psi})=\pi_{\bar{A}}(\ket{\psi})$ and
\begin{equation}
1/N_A\le\pi_{A}(\ket{\psi})\le 1.
\end{equation}
Moreover, $\pi_A(\ket{\psi})=1$ and $\pi_A(\ket{\psi})=1/N_A$ iff $\ket{\psi}$ is, respectively, separable and maximally entangled  with respect to the given bipartition.
\end{lemma}
\begin{proof}
The quadratic form $\pi_A(\ket{\psi})=\sum_l \lambda_l^2$ reaches its extremal values in the simplex $\Delta^{N_A}=\{(\lambda_l)_{l\in X^{n_A}} | 0\leq\lambda_l\leq 1, \sum_l \lambda_l=1\}$. The maximum is reached on the frontier, $\lambda_l=\delta_\ell(l)$ for some $1\leq\ell\leq n_A$, while the minimum is attained at the interior point where $\mathrm{d} \pi_A(\ket{\psi})=0$, i.e.\ $\lambda_l=1/N_A$. By  Theorems \ref{prop:separable} and \ref{prop:mbes} one gets the thesis.
\qed
\end{proof}
It follows that $\mathcal{L}_A(\ket{\psi})$ has a behavior similar to $\mathcal{E}_A(\ket{\psi})$. In particular,
\begin{theorem}[Linear entropy bounds]
\label{th:linentbounds}
One gets $0\leq\mathcal{L}_A(\ket{\psi})\leq 1$. Moreover $\mathcal{L}_A(\ket{\psi})=0$ iff $\ket{\psi}$ is separable with respect to the bipartition $(A,\bar{A})$,  while  $\mathcal{L}_A(\ket{\psi})= 1$ iff $\ket{\psi}$ is
maximally  entangled with respect to the bipartition $(A,\bar{A})$. \qed
\end{theorem}

\begin{remark}
\label{rem:idea}
Let us consider a system composed of an even number $n$ of qubits and a balanced bipartition $(A,\bar{A})$.
The information contained in a maximally bipartite entangled state $\ket{\psi}$ is not locally accessible by party $A$ or $\bar{A}$, because, by Corollary \ref{cor:mbe}, their partial density matrices are maximally mixed, $\rho_A=\rho_{\bar{A}}=1/N_A$. Rather, all information is totally shared by them. Note that if $n$ is odd, according to Lemma \ref{lem:reduced}, $\rho_{\bar{A}}$ cannot be maximally mixed. Rather, $\rho_{\bar{A}}=P/N_A$, where $P=1-\ket{v}\bra{v}$ is a codimension-1 projection, $\ket{v}$ being the normalized eigenvector belonging to the eigenvalue 0. Note that it is the constraint that the total system is in a pure state that prevents $\rho_{\bar{A}}$ from being of maximal rank.
\end{remark}
If the bipartition is not balanced, one gets
\begin{theorem}[Smaller subsystems]
\label{pr:subpartition}
A state $\ket{\psi}$ maximally entangled with respect to the bipartition $(A,\bar{A})$, is maximally entangled with respect to every bipartitions $(B,\bar{B})$ with $B\subset A$.
\end{theorem}
\begin{proof}
The Theorem is a consequence of Corollary \ref{cor:mbe} and the property that if subsystem $A$ has a maximally mixed density matrix, $\rho_A=1/N_A$, the density matrix of every smaller part $B\subset A$ is again maximally mixed,  $\rho_B=\tr_{\bar{B}\cap A}\rho_A=1/N_B$. \qed
\end{proof}
The explicit expressions of the reduced density matrix and its purity in terms of the Fourier coefficient of the state are given by the following
\begin{theorem}[Fourier expression of purity. Form 1]
\label{th:rhoA piA}
Given a bipartition $(A,\bar{A})$ and a state $\ket{\psi}\in\mathcal{H}_S$, one gets
\begin{equation}
\label{eq:rhoA}
\rho_A=  \sum_{k , l\in X^n}
z_{k} \bar{z}_{l}
\delta_{k_{\bar{A}},l_{\bar{A}}}
\ket{k_A} \bra{l_A}
\end{equation}
and
\begin{equation}
\pi_{A}(z)= \sum_{k,k',l,l' \in X^n} z_{k} z_{k'}
\bar{z}_{l}\bar{z}_{l'} \delta_{k_A, l'_A}
\delta_{k'_A,l_A} \delta_{k_{\bar{A}}, l_{\bar{A}}} \delta_{k'_{\bar{A}}, l'_{\bar{A}}}
\label{eq:puritygeneral},
\end{equation}
where $k_A=(k_i)_{i\in A}$, $\ket{l}_A=\bigotimes_{i\in A}\ket{l_i}_i\in\mathcal{H}_A$,
$\delta_{k, l}=\delta_{l,k}=\delta_{\{k\}}(l)$, and
\begin{equation}
\label{eq:normz}
z=(z_k)_{k\in X^n}\in\mathbb{S}^{2 N-1},\qquad  \mathbb{S}^{2N - 1}=\{z\in\mathbb{C}^{N} : \sum |z_k|^2=1\},
\end{equation}
with $N=2^n$, are the Fourier coefficients of $\ket{\psi}$ in the computational basis, introduced in Definition \ref{def:system}.
\end{theorem}
\begin{proof}
State $|\psi\rangle$  can be written accordingly to the bipartition $(A,\bar{A})$ as
\begin{equation*}
|\psi\rangle = \sum_{k\in X^n} z_k \ket{k_A}_A\otimes\ket{k_{\bar{A}}}_{\bar{A}} ,
\label{eq:genrandomxbip}
\end{equation*}
By plugging this expression into that of $\rho_A$ and $\pi_A$ given in Definition
\ref{def:puritydef} the results  follow.
\qed
\end{proof}
\begin{remark}
\label{rm:permutation}
Consider a \emph{reference bipartition} into two blocks of contiguous qubits  $(C,\bar{C})$,  namely $C=\{1,2,\dots, n_A\}$, then
\begin{equation}
\pi_{C}(z)= \sum_{l,l'\in X^{n_A}}\sum_{m,m'\in X^{n_{\bar{A}}}} z_{(l,m)} z_{(l', m')}
\bar{z}_{(l', m)} \bar{z}_{(l, m')},
\label{eq:purityref}
\end{equation}
where $(l,m)=(l_1,\dots,l_{n_A},m_1,\dots,m_{n_{\bar{A}}})\in X^n$.

Note that $A=p(C)$ for a suitable permutation $p$ of $S$.
In fact, there is a bijection,
\begin{equation}
\Phi: p\in\mathcal{P}_n^{n_A}  \mapsto (p(C),p(\bar{C})),
\end{equation}
between the subset
\begin{equation}
\mathcal{P}_n^{n_A}=\{p\in \mathcal{P}_n | p(i)<p(i+1), 1\leq i \leq n-1, i\neq n_A \}
\end{equation}
of the permutation group $\mathcal{P}_n$  and the set of all bipartitions $(A,\bar{A})$ of dimension $n_A$. We can write
\begin{eqnarray}
|\psi\rangle &=& \sum_{k\in X^n} z_k \ket{k_{p(C)}}_{p(C)}\otimes\ket{k_{p(\bar{C})}}_{p(\bar{C})}
\nonumber\\
&=& \sum_{l\in X^{n_A}} \sum_{m\in X^{n_{\bar{A}}}} z_{p^{-1}((l,m))}
|l\rangle_C\otimes|m\rangle_{\bar{C}} ,
\label{eq:genrandomxbip1}
\end{eqnarray}
whence, for $A=p(C)$,
\begin{eqnarray}
\pi_{A} (z) =
\sum_{l,l'\in X^{n_A}}\sum_{m,m'\in X^{n_{\bar{A}}}}  z_{p^{-1}(l,m)} z_{p^{-1}(l', m')}
\bar{z}_{p^{-1}(l', m)}   \bar{z}_{p^{-1}(l, m')}.
\label{eq:purityp}
\end{eqnarray}
\end{remark}
For generic bipartitions we have the following
\begin{corollary}[Purity. Form 2]
\label{cor:piA1}
\begin{eqnarray}
\pi_A(z)&=&\sum_{k,h\in X^n} z_k \, z_{k\oplus h} \, \bar{z}_{k \oplus h_A} \, \bar{z}_{k \oplus h_{\bar{A}}}\nonumber \\
&=&\sum_{k\in X^S} \sum_{l\in X^A}\sum_{m\in X^{\bar{A}}}z_k\,  z_{k \oplus l \oplus m} \,\bar{z}_{k \oplus l} \, \bar{z}_{k \oplus m}
\label{eq:piAoplus}
\end{eqnarray}
were $X^A$ and $X^{\bar{A}}$ are viewed as subspaces of $X^S$ with the natural injection, and $a\oplus b= (a_i \oplus b_i)_{i \in
S}= (a_i + b_i \mod 2)_{i \in
S}$ is the XOR operation.
\end{corollary}
\begin{proof}
By substituting in (\ref{eq:puritygeneral}) $k'=h\oplus k$, one gets
\begin{eqnarray*}
\pi_{A}(z)&=& \sum_{k,h,l,l' \in X^n} z_{k} z_{k\oplus h}
\bar{z}_{l}\bar{z}_{l'} \delta_{k_A, l'_A}
\delta_{k_A\oplus h_A,l_A} \delta_{k_{\bar{A}}, l_{\bar{A}}} \delta_{k_{\bar{A}}\oplus h_{\bar{A}}, l'_{\bar{A}}}\\
&=&\sum_{k,h,l,l' \in X^n} z_{k} z_{k\oplus h}
\bar{z}_{l}\bar{z}_{l'}
\delta_{k\oplus h_A,l} \delta_{k\oplus h_{\bar{A}}, l'}
\\
&=&\sum_{k,h \in X^n} z_{k} z_{k\oplus h}
\bar{z}_{k\oplus h_A}\bar{z}_{k\oplus h_{\bar{A}}},
\end{eqnarray*}
which is the first desired equality. The second equality follows by the identifications $l=h_A \in X^A$ and $m=h_{\bar{A}}\in X^{\bar{A}}$.
\qed
\end{proof}
\begin{remark}
The space $X^n$ is an $n$-dimensional vector space over the finite field $X=\mathbb{Z}_2$ with the standard addition and multiplication mod $2$. In this respect the XOR operation is nothing but the usual sum of  vectors of $X^S$ and $X^A$ and $X^{\bar{A}}$ are  vector subspaces.
\end{remark}
\begin{remark}
\label{rm:splitA}
Note that (\ref{eq:piAoplus}) can be split into three parts
\begin{eqnarray}
\pi_{A}(z)&=& \sum_{k\in X^n} |z_k|^4 + \sum_{k\in X^S} \sum_{l\in X^A_*} |z_k|^2 |z_{k\oplus l }|^2 + \sum_{k\in X^S} \sum_{m\in X^{\bar{A}}_*} |z_k|^2 |z_{k\oplus m }|^2
\nonumber\\
&+&\sum_{k\in X^S} \sum_{l\in X^A_*}\sum_{m\in X^{\bar{A}}_*} \Re\left[ z_k\,  z_{k \oplus l \oplus m} \,\bar{z}_{k \oplus l} \, \bar{z}_{k \oplus m}\right],
\label{eq:splitA}
\end{eqnarray}
where $X^A_*= X^A \backslash \{0\}$.

It is an easy exercise to check that the number of  monomials $|z_k|^4$ is $\mathcal{N}^{(1)}_{\mathrm{tot}}=2^n$,
the number of  monomials
 $|z_k|^2 |z_h|^2$ with $k\neq h$ is
\begin{equation}
\mathcal{N}^{(2)}_{\mathrm{tot}}=2^{n}\left(2^{n_A}+2^{n_{\bar{A}}}-2\right),
\end{equation}
and the number of monomials  $\Re[ z_k z_l \bar{z}_m \bar{z}_n]$ with distinct indices is
\begin{eqnarray}
 \mathcal{N}^{(4)}_{\mathrm{tot}}= 2^{n}
\left(2^{n_A}-1\right)\left(2^{n_{\bar{A}}}-1\right).
\end{eqnarray}
One gets $\mathcal{N}^{(1)}_{\mathrm{tot}}+\mathcal{N}^{(2)}_{\mathrm{tot}}+\mathcal{N}^{(4)}_{\mathrm{tot}}=2^{2n}$, in agreement with the first equality in (\ref{eq:piAoplus}). Moreover, the number of \emph{distinct} monomials of the various types are
$\mathcal{N}^{(1)}=\mathcal{N}^{(1)}_{\mathrm{tot}}$, $\mathcal{N}^{(2)}=\mathcal{N}^{(2)}_{\mathrm{tot}}/2$, and $\mathcal{N}^{(4)}=\mathcal{N}^{(4)}_{\mathrm{tot}}/4$.
\end{remark}

\section{Multipartite entanglement}

The aim of this section is to generalize the ideas of the previous section to the case of multipartite entanglement. We require that the information in a maximally multipartite entangled state be distributed as well as possible.
In the ideal case this would mean that
\begin{definition}[Perfect MMES]
\label{def:perfect}
A  state $\ket{\psi}$ maximally entangled with respect to every bipartition $(A,\bar{A})$ is called a \emph{perfect maximally multipartite entangled state (perfect MMES)}.
\end{definition}
\begin{theorem}[Perfect MMES characterization]
\label{prop:perfect}
The following statements are equivalent:
\begin{enumerate}
\item $\ket{\psi}\in\mathcal{H}_S$ is a perfect MMES;
\item $\rho_A=1 /N_A$ for every subsystem $A\subset S$ with $n_A \leq n/2$;
\item $\rho_A=1 /N_A$ for every maximal subsystem $A\subset S$; 
\item $\mathcal{E}_A(\ket{\psi})= \mathcal{L}_A(\ket{\psi})= 1$ for every balanced bipartition $(A,\bar{A})$;
\item $\pi_A(\ket{\psi})= 1/N_A$ for every balanced bipartition $(A,\bar{A})$.
\end{enumerate}
\end{theorem}
\begin{proof}
Equivalence between 1 and 2 follows from Definition \ref{def:perfect} and Corollary \ref{cor:mbe}. Statements 2 and 3 are equivalent by Theorem \ref{pr:subpartition}. Equivalence between 3 and 4 follows from Definition \ref{def:mbes} and Theorem \ref{th:linentbounds}. Finally, 4 and 5 are equivalent by virtue of Lemma \ref{lm:purityconstraint}.
\qed
\end{proof}
\begin{remark}
Note that the requirement that a given balanced bipartition $(A,\bar{A})$ be in a maximally entangled state could collide with the same requirement for a different balanced bipartition $(B,\bar{B})$, with $B\neq A$. Indeed, the local unitaries $U^A$ and $U^{\bar{A}}$ in Theorem \ref{prop:mbes} are in general nonlocal for the bipartition $(B,\bar{B})$. Thus, at variance with the bipartite case,  a perfect MMES cannot exist. This motivates the following definition.
\end{remark}
\begin{definition}[MMES]
\label{def:MMES} Let us define the \emph{potential of multipartite entanglement} as
\begin{equation}
\label{eq:pime}
\pi_{\mathrm{ME}}(\ket{\psi})
=\left( \begin{array}{c}  n \\  {[n/2]}  \end{array}  \right)^{-1}
\sum_{|A|=[n/2]}\pi_A(\ket{\psi}).
\end{equation}
A \emph{maximally multipartite entangled state (MMES)} $\ket{\varphi}$ is a minimizer of $\pi_{\mathrm{ME}}$,
\begin{equation}
\label{eq:minimizer}
\pi_{\mathrm{ME}}(\ket{\varphi})=\min \{\pi_{\mathrm{ME}}(\ket{\psi})\; | \; \ket{\psi}\in \mathcal{H}_S, \bra{\psi}\psi\rangle=1\}.
\end{equation}
\end{definition}
The potential $\pi_{\mathrm{ME}}$  measures the average bipartite entanglement over all possible  $\left( \begin{array}{c}  n \\  {[n/2]}  \end{array}  \right)$ balanced bipartition and thus inherits the bounds on the purity given in Lemma \ref{lm:purityconstraint}, namely,
\begin{lemma}[Bounds on $\pi_{\mathrm{ME}}$]
\label{lm:pmeconstraint}
The potential of multipartite entanglement satisfies
\begin{equation}
1/N_A\le\pi_{\mathrm{ME}}(\ket{\psi})\le 1,
\end{equation}
with $N_A=2^{[n/2]}$, for all normalized $\ket{\psi}\in\mathcal{H}_S$. \qed
\end{lemma}
The upper and lower bounds are characterized by the following
\begin{theorem}[Optimizing states]
\label{th:fullyfactorized}
The upper bound $\pi_{\mathrm{ME}}(z)=1$ is attained by the \emph{fully factorized states}, whose Fourier coefficients in the computational basis $z=(z_k)_{k\in X^n}$ are $z_k=\prod_{i\in S} \alpha_{k_i}^i$, with $|\alpha_0^i|^2 + |\alpha_1^i|^2=1$. On the other hand, the lower bound $\pi_{\mathrm{ME}}(z)=1/N_A$, if attained,  would correspond to a perfect MMES.
\end{theorem}
\begin{proof}
$\pi_{\mathrm{ME}}(\ket{\psi})= 1$ iff $\pi_{A}(\ket{\psi})=1$ for all balanced bipartitions $(A,\bar{A})$. By Lemma \ref{lm:purityconstraint} this happens iff $\ket{\psi}$ is separable with respect to all balanced bipartitions. Now, note that $\ket{\psi}=\ket{v_1}_A\otimes\ket{v_2}_{\bar{A}}$ and $\ket{\psi}=\ket{v_3}_B\otimes\ket{v_4}_{\bar{B}}$
iff $\ket{\psi}=\ket{v_{13}}_{A\cap B}\otimes\ket{v_{14}}_{A\cap\bar{B}}\otimes\ket{v_{23}}_{\bar{A}\cap B}\otimes\ket{v_{24}}_{\bar{A}\cap\bar{B}}$. Since for all $i\in S$, $\{i\}=\cap_r A_r$  for a suitable set $\{A_r\}$ of maximal subsystems, 
one has $\ket{\psi}=\bigotimes_{i\in S} \ket{v_i}_i$ with $\bra{v_i}v_i\rangle=1$. Thus $\ket{\psi}=\sum_{k\in X^{n}}\ket{k}\prod_{i\in S} \bra{k_i}v_i\rangle$, and the first part of the theorem follows by setting $\alpha_{k_i}^i=\bra{k_i}v_i\rangle$.
Concerning the second part, $\pi_{\mathrm{ME}}(\ket{\psi})=1/N_A$ iff $\pi_A(\ket{\psi})=1/N_A$ for all balanced bipartitions $(A,\bar{A})$. By Theorem \ref{prop:perfect} this happens iff $\ket{\psi}$ is a perfect MMES.
\qed
\end{proof}
\begin{remark}
In words, a perfect MMES is characterized by a multipartite entanglement that is maximum, in the sense that it saturates the minimum of the purity \emph{and} such a minimum does not depend on the bipartition.
However, if the minimum of the potential of multipartite entanglement is strictly larger than the lower bound in Lemma \ref{lm:pmeconstraint}, i.e.\ $\min \pi_{\mathrm{ME}}>1/N_A$, it may happen that different bipartitions yield different values of $\pi_A$, some of them smaller than $\min \pi_{\mathrm{ME}}$, some larger.
In such a situation, one can seek those states among the minimizers that have the lowest variance. This quest can be recast as an optimization problem \cite{MMES}. We will not elaborate further on this issue.
\end{remark}
Now we will examine in more details the potential of multipartite entanglement and we will determine its form.
\begin{theorem}[Fourier expression of $\pi_{\mathrm{ME}}$. Form 1]
\label{th:piDelta}
Given a state $\ket{\psi}\in\mathcal{H}_S$, the potential of multipartite entanglement has the following expression in terms of its Fourier coefficients in the computational basis $z=(z_k)_{k\in X^n}$
\begin{eqnarray}
\pi_{\mathrm{ME}}(z)
=  \sum_{k,k',l,l' \in X^n} \Delta(k, k'; l, l';[n/2])\, z_{k}\, z_{k'}\,
\bar{z}_{l}\, \bar{z}_{l'}\, ,
\label{eq:pimeDelta}
\end{eqnarray}
with a \emph{coupling function}
\begin{eqnarray}
\Delta(k,k';l,l';n_A)
=\frac{1}{2} \tilde\Delta(k,k';l,l';n_A) +\frac{1}{2} \tilde\Delta(k',k;l,l';n_A).
\label{eq:Deltadef1}
\end{eqnarray}
where
\begin{eqnarray}
\tilde \Delta(k,k';l,l';n_A)
=\left(\begin{array}{c}n \\n_A\end{array}\right)^{-1} \sum_{|A|=n_A} \delta_{k_{A}, l'_{A}} \delta_{k'_{A},l_{A}}
\delta_{k_{\bar{A}}, l_{\bar{A}}} \delta_{k'_{\bar{A}}, l'_{\bar{A}}} .
\label{eq:Deltadef}
\end{eqnarray}
\end{theorem}
\begin{proof}
The result follows by plugging the expression (\ref{eq:puritygeneral}) of $\pi_A$ given by Theorem \ref{th:rhoA piA} into (\ref{eq:pime}) of Definition \ref{def:MMES}, and by symmetrizing.
\qed
\end{proof}
\begin{remark}
In the spirit of Remark \ref{rm:permutation}, it is easy to see that the average can be extended to the whole permutation group, yielding
\begin{eqnarray}
\tilde\Delta(k,k';l,l';|C|)
=\frac{1}{n!}\sum_{p\in\mathcal{P}_n}  \delta_{k_{p(C)}, l'_{p(C)}} \delta_{k'_{p(C)},l_{p(C)}}
\delta_{k_{p(\bar{C})}, l_{p(\bar{C})}} \delta_{k'_{p(\bar{C})}, l'_{p(\bar{C})}} .
\end{eqnarray}
\end{remark}
\begin{remark}
Note that $\tilde\Delta$ would have served as well as $\Delta$ as a coupling function, namely
\begin{equation}
\pi_{\mathrm{ME}}(z)
=  \sum_{k,k',l,l' \in X^n} \tilde \Delta(k, k'; l, l';[n/2])\, z_{k}\, z_{k'}\,
\bar{z}_{l}\, \bar{z}_{l'}\, .
\end{equation}
However, while
$$\tilde \Delta(k,k';l,l';n_A)= \tilde \Delta(l,l';k,k';n_A),$$
which ensures the reality of $\pi_{\mathrm{ME}}$, one gets
$$\tilde \Delta(k',k;l,l';n_A)= \tilde \Delta(k,k';l',l ;n_A)=\tilde \Delta(k,k';l',l ;n_{\bar{A}}).$$
Thus, $\tilde \Delta(k, k'; l, l';[n/2])$ is a symmetric function of the pairs $(k,k')$ and $(l,l')$ only
if $n$ is even, when $[n/2]=n/2=n_A= n_{\bar{A}}$. Since $\pi_{\mathrm{ME}}$ does not depend on the antisymmetric
part of the coupling function, we shall use the symmetric coupling function $\Delta$.
We summarize its properties, which easily derive from this Remark in the following
\end{remark}
\begin{lemma}[Coupling function symmetries]
The coupling function $\Delta: X^{2 n}\times X^{2 n}\times  \mathbb{N}\to\mathbb{Q}$  has the following symmetries
\begin{equation}
\Delta(k,k';l,l';n_A)= \Delta(k',k; l,l';n_A)= \Delta(l,l'; k,k';n_A),
\end{equation}
for every  $k, k', l, l' \in X^{n}$ and every $n_A$ with $1\leq n_A \leq n-1$. \qed
\end{lemma}
The following definition and the subsequent lemma are the main ingredients for determining the explicit expression of the coupling function $\Delta$.
\begin{definition}[Admissible set]
Let us define the \emph{admissible set} as the set of all quadruples of sequences that yield a nonvanishing contribution to the  function $\tilde \Delta$, that is
\begin{eqnarray}
Q_{n_A}&=&\{(k,k',l,l')\in X^{4n} | k_A=l'_A,\, k'_A=l_A, \text{ and } k_{\bar{A}}=l_{\bar{A}},\, k'_{\bar{A}}=l'_{\bar{A}},
\nonumber\\
& &\phantom{\{(k,k',l,l')\in X^{4n} |} \text{ for some } (A,\bar{A}) \text{ with } |A|=n_A\}.
\end{eqnarray}
\end{definition}
Obviously,
\begin{equation}
Q_{n_A}\subset Q=\bigcup_{0\leq s\leq n} Q_{s}.
\end{equation}
\begin{lemma}[Admissible set characterization]
\label{lm:ker}
The set Q is the kernel of the function $q: X^{4n}\to X^n$,
\begin{equation}
q(k,k',l,l')=\big((k\oplus l) \vee (k' \oplus l')\big) \wedge \big((k \oplus l') \vee (k'\oplus l)\big),
\end{equation}
where $a\oplus b= (a_i \oplus b_i)_{i \in
S}= (a_i + b_i \mod 2)_{i \in
S}$ is the XOR operation, $a\vee b=(a_i \vee b_i)_{i\in S}=(a_i + b_i + a_i b_i \mod 2)_{i\in S}$
the OR operation, and $a\wedge b=(a_i \wedge b_i)_{i\in S}=(a_i b_i)_{i\in S}$ the AND
operation.
\end{lemma}
\begin{proof}
The proof consists in a straightforward application of the above defined binary operations:
\begin{eqnarray}
Q &=&\{(k,k',l,l') | k_A=l'_A,\, k'_A=l_A,\, k_{\bar{A}}=l_{\bar{A}},\, k'_{\bar{A}}=l'_{\bar{A}}, \text{ for some } A\subset S\}\nonumber\\
&=&\{ k_i=l'_i,\, k'_i=l_i,\, k_j=l_j,\, k'_j=l'_j, \text{ with } i \in A, j \in \bar{A}\}\nonumber\\
&=&\{ k_i \oplus l'_i=0,\, k'_i\oplus l_i=0,\, k_j\oplus l_j=0,\, k'_j \oplus l'_j, \text{ with } i \in A, j \in \bar{A}\}\nonumber\\
&=&\{ (k_i \oplus l'_i) \vee (k'_i\oplus l_i)=0,\, (k_j\oplus l_j) \vee (k'_j \oplus l'_j)=0,  i \in A, j \in \bar{A}\}\nonumber\\
&=&\{ \big((k_i \oplus l'_i) \vee (k'_i\oplus l_i)\big) \wedge \big((k_i\oplus l_i) \vee (k'_i \oplus l'_i)\big)=0,  i \in S\}\nonumber\\
&=&\{ \big((k \oplus l') \vee (k' \oplus l)\big) \wedge \big((k \oplus l) \vee (k' \oplus l')\big)=0\}
\nonumber\\
&=& \mathop{\text{ker}} q . \nonumber
\end{eqnarray}
\qed
\end{proof}
\begin{remark}
Note that  $X^n$ can be viewed as a product ring  (of $n$ copies of $X=\mathbb{Z}_2$) with the  addition and multiplication mod $2$ defined componentwise, as usual. In this respect, the XOR operation is the sum $a+b$ and the AND operation is the product $a \cdot b$ of elements $a$ and $b$ of the product ring $X^n$. The OR operation is nothing but $a+b+a\cdot b$.
\end{remark}

After having proven all preparatory lemmata, now we come to the main result of this section that establishes an explicit form for the coupling function of the potential of multipartite entanglement.
\begin{theorem}[Coupling function]
\label{th:Delta}
The coupling function $\Delta$ has the following expression
\begin{equation}
\Delta(k,k';l,l';n_A)= g\big((k\oplus l) \vee (k' \oplus l'), (k \oplus l') \vee (k'\oplus l);n_A\big),
\end{equation}
where
\begin{equation}
 g(a,b;n_A) = \delta_0(a\wedge b) \; \hat{g}(|a|,|b|;n_A),
\label{eq:gdef}
\end{equation}
with $|a|=\sum_{i\in S} a_i$, and
\begin{eqnarray}
\hat{g}(s,t ;n_A)=
\frac{1}{2}\left(\begin{array}{c}n \\{n_A}\end{array}\right)^{-1}  \left[ \left(\begin{array}{c}
     n-s-t    \\
      {n_A}-s
\end{array} \right) +\left(\begin{array}{c}
     n-s-t    \\
      {n_A}-t
\end{array} \right)\right] .
\label{eq:hatg}
\end{eqnarray}
\end{theorem}
\begin{proof}
Let
$$a=(k\oplus l) \vee (k' \oplus l') \quad \text{and}\quad  b=(k \oplus l') \vee (k'\oplus l).$$
By Lemma \ref{lm:ker}, $(k,k',l,l')\in Q$ iff $a\wedge b=0$, and thus
$S_1=\{i\in S | a_i=b_i=1\}=\emptyset$. Therefore,
$$S=S_0+ A_1+B_1,$$
where $S_0=\{i\in S | a_i=b_i=0\}$, $A_1=\{i\in S | a_i=1\}$,  and $B_1=\{i \in S | b_i=1\}$.
Moreover,  it is easy to see that $a_i =0$ iff $k_i=l_i$ and $k'_i=l'_i$, with $i\in S$.
Thus,  if $(k,k',l,l')\in Q$, then $k_{\bar{A}_1}=l_{\bar{A}_1}$ and $k'_{\bar{A}_1}=l'_{\bar{A}_1}$, and, analogously,  $k_{\bar{B}_1}=l'_{\bar{B}_1}$ and $k'_{\bar{B}_1}=l_{\bar{B}_1}$.
As a consequence, $(k,k',l,l')\in Q_{n_A}$ iff there is a bipartition $(A,\bar{A})$, with $|A|=n_A$, such that
$$\bar{A}\subset \bar{A}_1 \quad \text{and}\quad A \subset \bar{B}_1,$$
that is
$A_1 \subset A$ and $B_1 \subset \bar{A}$. In other words, $(k,k',l,l')\in Q_{n_A}$ iff
$(k,k',l,l')\in Q$ and $|A_1|=|a|\leq n_A$, $|B_1|=|b|\leq n_{\bar{A}}$. Therefore, we can write
$$(k,k',l,l')\in Q_{n_A} \quad\text{iff} \quad a\wedge b=0, \text{ with }  |a|\leq n_A, \, |b|\leq n_{\bar{A}} ,$$
whence
$$\tilde\Delta(k,k';l,l';n_A)=\delta_0(a\wedge b)\, \delta_{[0,n_A]}(|a|)\,  \delta_{[0,n_{\bar{A}}]}(|b|)\, \left(\begin{array}{c}n \\{n_A}\end{array}\right)^{-1} \#(k,k',l,l'),$$
where
$\#(k,k',l,l')$ is the number of terms  of  the sum (\ref{eq:Deltadef}) that contribute to the  function $\tilde\Delta$ in Theorem \ref{th:piDelta}.

Now, according to the above conclusions, for a given admissible quadruple $(k,k',l,l')\in Q_{n_A}$  the number of terms $\#(k,k',l,l')$ is given by the number of bipartitions $(A,\bar{A})$ with $|A|=n_A$ and with
$A\subset \bar{B}_1=A_1+S_0$ and $\bar{A} \subset \bar{A}_1 =B_1+S_0$.
Since $A\cap\bar{A}=\emptyset$, $A_1\subset A$ and $B_1 \subset \bar{A}$, parties $A$ and $\bar{A}$ contend only for $S_0=S_0\cap A + S_0\cap \bar{A}$, namely
$$A=A_1+S_0\cap A \quad \text{and}\quad \bar{A} = B_1+S_0\cap \bar{A}.$$
Thus, their number equals the number of ways that $|A\backslash A_1|$ objects can be chosen from among $|S_0|$ objects. But $|A\backslash A_1|=|A|-|A_1|=n_A- |a|$ and $|S_0|=|S|-|A_1|-|B_1|=n-|a|-|b|$. Therefore,
$$\#(k,k',l,l')=\left(\begin{array}{c}
     n-|a|-|b|    \\
      {n_A}-|a|
\end{array} \right) .$$
By putting all together, and by stipulating that the binomial coefficient is zero when its arguments are negative, we obtain the stated form of the functions $\tilde\Delta$ and its symmetric part $\Delta$.
\qed
\end{proof}
\begin{remark}
\label{th:Delta1}
It is not difficult to see that an alternative form of $\hat{g}$ is the following
\begin{eqnarray}
\hat{g}(s,t ;n_A)=
\frac{1}{2}\left(\begin{array}{c}n \\ s, t\end{array}\right)^{-1}  \left[ \left(\begin{array}{c}
     {n_A}    \\
      s
\end{array} \right)\left(\begin{array}{c}
     {n_{\bar{A}}}    \\
      t
\end{array} \right) +\left(\begin{array}{c}
     {n_A}    \\
      t
\end{array} \right)\left(\begin{array}{c}
     {n_{\bar{A}}}    \\
      s
\end{array} \right)\right] ,
\label{eq:hatg1}
\end{eqnarray}
where
$$\left(\begin{array}{c}n \\ s, t\end{array}\right)= \frac{n!}{s!\, t!\, (n-s-t)!}$$
is the multinomial coefficient.
\end{remark}

By using the explicit form of the coupling function $\Delta$ one can give  the potential of multipartite entanglement a different form that has the advantage of being a sum over three indices only.
\begin{theorem}[$\pi_{\mathrm{ME}}$. Form 2]
\label{th:pme1}
The potential of multipartite entanglement can be written as
\begin{equation}
\pi_{\mathrm{ME}}(z)=\sum_{k,l,m\in X^n} g(l,m;[n/2])\, \Re\left[z_k\, z_{k\oplus l\oplus m}\, \bar{z}_{k\oplus l}\,\bar{z}_{k\oplus m}\right].
\label{eq:pimeoplus}
\end{equation}
\end{theorem}
\begin{proof}
Since $k\oplus 0=k$,
\begin{eqnarray*}
Q &=&\{(k,k',l,l') | k_A=l'_A,\, k'_A=l_A,\, k_{\bar{A}}=l_{\bar{A}},\, k'_{\bar{A}}=l'_{\bar{A}}, \text{ for some } A\subset S\}\\
&=&\{k_A=l'_A,\, k'_A\oplus l_A=0,\, k_{\bar{A}}=l_{\bar{A}},\, k'_{\bar{A}}\oplus l'_{\bar{A}}=0\}\\
&=&\{k_A=l'_A\oplus k'_A\oplus l_A,\, k'_A\oplus l_A=0,\, k_{\bar{A}}=l_{\bar{A}}\oplus k'_{\bar{A}}\oplus l'_{\bar{A}},\, k'_{\bar{A}}\oplus l'_{\bar{A}}=0\}\\
&=&\{k=k'\oplus l\oplus l',\, k'_A\oplus l_A=0,\,  k'_{\bar{A}}\oplus l'_{\bar{A}}=0\}\\
&=&\{k=k'\oplus l\oplus l',\, (k'\oplus l) \wedge (k' \oplus l')=0\}.
\end{eqnarray*}
Moreover, since $k\oplus k=0$, substituting for $k=k'\oplus l\oplus l'$ one gets
 $a= (k\oplus l) \vee (k' \oplus l')=  k' \oplus l'$
 and $b=(k\oplus l') \vee (k' \oplus l)=k' \oplus l$. Therefore,
$$\Delta(k,k';l,l';n_A)= \delta_{k,k'\oplus l\oplus l'}\,
g(k' \oplus l', k'\oplus l;n_A),$$
whence
$$
\pi_{\mathrm{ME}}(z)
=  \sum_{k',l,l' \in X^n} g(k' \oplus l', k'\oplus l;[n/2])\, z_{k'\oplus l\oplus l'}\, z_{k'}\,
\bar{z}_{l}\, \bar{z}_{l'}\, .
$$
By setting $l'=l\oplus k'$ and $l=m\oplus k'$, one obtains
$$
\pi_{\mathrm{ME}}(z)
=  \sum_{k',l,m \in X^n} g(l, m;[n/2])\, z_{k'\oplus l \oplus m}\, z_{k'}\,
\bar{z}_{k'\oplus m}\, \bar{z}_{k'\oplus l}\, .
$$
The thesis follows from the reality of $\pi_{\mathrm{ME}}(z)$.
\qed
\end{proof}
In analogy with the bipartite case examined in Remark \ref{rm:splitA}, the sum in (\ref{eq:pimeoplus}) can be split into three terms.
\begin{corollary}[$\pi_{\mathrm{ME}}$. Form 3]
\label{th:pme2}
\begin{eqnarray}
\pi_{\mathrm{ME}}(z)&=& \sum_{k\in X^n} |z_k|^4 + 2 \sum_{k\in X^n} \sum_{l\in X^n_*} \hat{g}(|l|,0;[n/2])\, |z_k|^2 |z_{k\oplus l }|^2
\nonumber\\
&+&\sum_{k\in X^n} \sum_{l,m\in X^n_*}
 g(l,m;[n/2])\, \Re\left[z_k\, z_{k\oplus l\oplus m}\, \bar{z}_{k\oplus l}\,\bar{z}_{k\oplus m}\right],
 \label{eq:splitpme}
 \end{eqnarray}
where $X^n_*= X^n \backslash \{0\}$.
\end{corollary}
\begin{proof}
The monomials $|z_k|^4$ are obtained from (\ref{eq:pimeoplus}) when $l=m=0$. In such a case $g(0,0;[n/2])=\hat{g}(0,0;[n/2])=1$. On the other hand,  the monomials $|z_k|^2 |z_h|^2$ with $k\neq h$ are obtained when either $l=0$ or $m=0$. In such a case, since $\delta_0(l\wedge 0)=1$ for all $l\in X^n$,  $g(l,0;[n/2])=g(0,l;[n/2])=\hat{g}(|l|,0;[n/2])$.
\qed
\end{proof}
A measure of the complexity of the potential of multipartite entanglement is given by the number of its terms. In particular, as we will see in the following, the crucial ones are the interfering monomials $\Re[ z_k z_l \bar{z}_m \bar{z}_n]$.
\begin{theorem}[Number of terms in $\pi_{\mathrm{ME}}$]
\label{th:pme3e}
Consider $\pi_{\mathrm{ME}}(z)$. The number of distinct monomials $|z_k|^4$  and the number of distinct monomials
 $|z_k|^2 |z_h|^2$ with $k\neq h$ are
\begin{equation}
\mathcal{N}^{(1)}=2^n, \qquad
\mathcal{N}^{(2)}=2^{2n-2}-2^{n-1}+ \frac{2^{n}}{3+(-1)^n}
\left(\begin{array}{c} n    \\ {[n/2]} \end{array} \right),
\label{eq:N1 N2}
\end{equation}
respectively. The number of distinct monomials  $\Re[ z_k z_l \bar{z}_m \bar{z}_n]$ with distinct indices is
\begin{eqnarray}
 \mathcal{N}^{(4)}= 2^{n-3} \sum_{1\leq s, t \leq \left[\frac{n+1}{2}\right]}
\left(\begin{array}{c} n  \\ s \end{array} \right)
\left(\begin{array}{c} n-s  \\ t \end{array} \right)
= 2^{n-3} \sum_{1\leq s, t \leq \left[\frac{n+1}{2}\right]}
\left(\begin{array}{c} n  \\ s, t \end{array} \right).
\label{eq:N4}
\end{eqnarray}
\end{theorem}
\begin{proof}
The total number of terms of the sum in (\ref{eq:pimeoplus}) is given by
\begin{eqnarray*}
 \mathcal{N}_{\mathrm{tot}}&=&\sum_{k,l,m\in X^n} \delta_{\mathbb{Q}_*}\big(g(l,m;[n/2])\big)\\
&=&
2^n \sum_{l,m\in X^n} \delta_0(l\wedge m)\, \delta_{\left[0,\frac{n+1}{2}\right]}(|l|)\,
\delta_{\left[0,\frac{n+1}{2}\right]}(|m|)\\
&=& 2^n \sum_{0\leq s, t \leq \left[\frac{n+1}{2}\right]} \sum_{l,m\in X^n} \delta_0(l\wedge m)\,
\delta_s(|l|)\, \delta_t(|m|) \\
&=& 2^n \sum_{0\leq s, t \leq \left[\frac{n+1}{2}\right]} \sum_{l\in X^n}
\delta_s(|l|)  \sum_{m\in X^{n-s}} \delta_t (|m|) \\
&=& 2^n \sum_{0\leq s, t \leq \left[\frac{n+1}{2}\right]}
\left(\begin{array}{c} n  \\ s \end{array} \right)
\left(\begin{array}{c} n-s  \\ t \end{array} \right) .
\end{eqnarray*}
Therefore, the total number of monomials $|z_k|^4$ is
\begin{eqnarray*}
 \mathcal{N}_{\mathrm{tot}}^{(1)}&=&\sum_{k\in X^n} \delta_{\mathbb{Q}_*}\big(g(0,0;[n/2])\big) = 2^n
\left(\begin{array}{c} n  \\ 0 \end{array} \right)
\left(\begin{array}{c} n \\ 0 \end{array} \right)
= 2^n,
\end{eqnarray*}
while the total number of  monomials
 $|z_k|^2 |z_h|^2$ with $k\neq h$ is
\begin{eqnarray*}
 \mathcal{N}_{\mathrm{tot}}^{(2)}&=&2 \sum_{k\in X^n}\sum_{l\in X^n_*} \delta_{\mathbb{Q}_*}\big(g(l,0;[n/2])\big) = 2^{n+1}
 \sum_{1\leq t \leq \left[\frac{n+1}{2}\right]}
\left(\begin{array}{c} n  \\ 0 \end{array} \right)
\left(\begin{array}{c} n  \\ t \end{array} \right)\\
&=& 2^{n} \sum_{1\leq t \leq \left[\frac{n+1}{2}\right]}
\left[\left(\begin{array}{c} n  \\ t \end{array} \right)+\left(\begin{array}{c} n  \\ {n-t} \end{array} \right)\right]\\
&=& 2^{n}  \sum_{1\leq t \leq \left[\frac{n+1}{2}\right]}
\left(\begin{array}{c} n  \\ t \end{array} \right)+ 2^n
 \sum_{\left[\frac{n}{2}\right]\leq t \leq n}
\left(\begin{array}{c} n  \\ t \end{array} \right)\\
&=& 2^{n}  \sum_{1\leq t \leq n}
\left(\begin{array}{c} n  \\ t \end{array} \right)+ 2^n
 \sum_{\left[\frac{n}{2}\right]\leq t \leq \left[\frac{n+1}{2}\right]}
\left(\begin{array}{c} n  \\ t \end{array} \right)\\
&=& 2^{2 n} -2^n + 2^n
 \sum_{\left[\frac{n}{2}\right]\leq t \leq \left[\frac{n+1}{2}\right]}
\left(\begin{array}{c} n  \\ {[n/2]} \end{array} \right)\\
&=& 2^{2 n} -2^n +
\frac{2^{n+2}}{3+(-1)^n}
\left(\begin{array}{c} n  \\ {[n/2]} \end{array} \right).
\end{eqnarray*}
On the other hand, the total number of  monomials  $\Re[ z_k z_l \bar{z}_m \bar{z}_n]$ with distinct indices reads
\begin{eqnarray*}
 \mathcal{N}_{\mathrm{tot}}^{(4)}= \sum_{k\in X^n}\sum_{l,m\in X^n_*} \delta_{\mathbb{Q}_*}\big(g(l,m;[n/2])\big) =
2^n \sum_{1\leq s, t \leq \left[\frac{n+1}{2}\right]}
\left(\begin{array}{c} n  \\ s \end{array} \right)
\left(\begin{array}{c} n-s  \\ t \end{array} \right).
\end{eqnarray*}
The results follow, since by symmetry, the numbers of distinct monomials are
$\mathcal{N}^{(1)}=\mathcal{N}_{\mathrm{tot}}^{(1)}$, $\mathcal{N}^{(2)}=\mathcal{N}_{\mathrm{tot}}^{(2)}/4$, and $\mathcal{N}^{(4)}=\mathcal{N}_{\mathrm{tot}}^{(4)}/8$.
\qed
\end{proof}
\begin{table}
\caption{Number of monomials in $\pi_{\mathrm{ME}}(z)$.}
\label{tab:monomials}       
\centering
\begin{tabular}{llll}
\hline\noalign{\smallskip}
$n$ & $\mathcal{N}^{(1)}$ & $\mathcal{N}^{(2)}$ & $\mathcal{N}^{(4)}$  \\
\noalign{\smallskip}\hline\noalign{\smallskip}
2 & 4 & 4 & 1\\
3 & 8 & 24 & 12\\
4 & 16 & 80 & 84\\
5 & 32 & 400 & 680\\
6 & 64 & 1312 & 4000\\
7 & 128 & 6272 & 28672\\
8 & 256 & 20736 & 162624\\
$n\to\infty$ & $2^{n}$  & $2^{2n-2}$ & $2^{n-3}3^n$\\
\noalign{\smallskip}\hline
\end{tabular}
\end{table}
\begin{remark}
For large values of $n$, by making use of Stirling's approximation one gets
$$\left(\begin{array}{c} n  \\ n/2 \end{array} \right) \sim  2^n \sqrt{\frac{2}{\pi n}},$$
hence, from (\ref{eq:N1 N2})
\begin{equation}
\mathcal{N}^{(2)}\sim 2^{2n-2}\left(1+ \frac{4}{3+(-1)^n} \sqrt{\frac{2}{\pi n}}
 \right), \qquad n\to\infty.
\end{equation}
The asymptotics of $\mathcal{N}^{(4)}$ is a little more elaborated.
 First note that, by Stirling,
\begin{eqnarray*}
 \mathcal{N}^{(4)} \sim 2^{n-3} \sum_{1\leq s, t \leq \left[\frac{n+1}{2}\right]}
\frac{1}{2\pi n\sqrt{\frac{s}{n} \frac{t}{n}\left(1-\frac{s}{n}-\frac{t}{n}\right)}}
\exp\left( n\, H\left(\frac{s}{n},\frac{t}{n}\right)\right),
\end{eqnarray*}
where the function $H: \Delta^2 \to \mathbb{R}$, defined on the simplex $\Delta^2=\{(x,y)\in [0,1]^2\, | \, x+y=1 \}$,
is the entropy
\begin{equation*}
H(x,y)=-x \log x - y \log y - (1-x-y) \log(1-x-y).
\end{equation*}
Then, for $n\to\infty$, by using the same arguments as in the proof of Laplace - De Moivre theorem \cite{sinai}, one can show that
\end{remark}
\begin{eqnarray}
\mathcal{N}^{(4)} \sim 2^{n-3} 3^n, \qquad n\to\infty.
\end{eqnarray}
The numbers of different types of monomials appearing in $\pi_{\mathrm{ME}}(z)$, as well as their asymptotic expansions, are given in Table \ref{tab:monomials}.

We conclude the section by exhibiting another form of multipartite entanglement.
\begin{theorem}[$\pi_{\mathrm{ME}}$. Form 4]
\label{th:pime4}
The potential of multipartite entanglement can be written as
\begin{eqnarray}
\pi_{\mathrm{ME}}(z)= 1-\frac{1}{2}\sum_{k,l,m\in X^n}
 g(l,m;[n/2])\, \left|z_k\, z_{k\oplus l\oplus m} - z_{k\oplus l}\, z_{k\oplus m}\right|^2.
 \label{eq:pime4}
 \end{eqnarray}
\end{theorem}
\begin{proof}
Let us consider (\ref{eq:pimeoplus}). By substituting the identity
\begin{eqnarray*}
\Re\left[z_k\, z_{k\oplus l\oplus m}\, \bar{z}_{k\oplus l}\,\bar{z}_{k\oplus m}\right] &=&
-\frac{1}{2} \left|z_k\, z_{k\oplus l\oplus m} - z_{k\oplus l}\, z_{k\oplus m}\right|^2
\nonumber\\
& & + \frac{1}{2} \left|z_k\, z_{k\oplus l\oplus m}\right|^2 +  \frac{1}{2} \left| z_{k\oplus l}\, z_{k\oplus m}\right|^2,
\end{eqnarray*}
one gets
\begin{eqnarray*}
\pi_{\mathrm{ME}}(z)&=& -\frac{1}{2}\sum_{k,l,m\in X^n}
 g(l,m;[n/2])\, \left|z_k\, z_{k\oplus l\oplus m} - z_{k\oplus l}\, z_{k\oplus m}\right|^2
 \nonumber\\
 & &+  \frac{1}{2} \sum_{k,l,m\in X^n} g(l,m;[n/2])\,\left(\left|z_k\, z_{k\oplus l\oplus m}\right|^2 +  \left| z_{k\oplus l}\, z_{k\oplus m}\right|^2\right).
\end{eqnarray*}
Now, by simple manipulations,
$$ \sum_{k\in X^n}  \left| z_{k\oplus l}\, z_{k\oplus m}\right|^2= \sum_{k\in X^n}\left|z_k\, z_{k\oplus l\oplus m}\right|^2$$
and
$$\sum_{k,l,m\in X^n} g(l,m;[n/2])\,\left|z_k\, z_{k\oplus l\oplus m}\right|^2 =
\sum_{k,l \in X^n} \left|z_k\, z_{k\oplus l}\right|^2 \sum_{m\in X^n} g(l\oplus m,m;[n/2]).
$$
Thus,
\begin{eqnarray*}
\pi_{\mathrm{ME}}(z)&=& -\frac{1}{2}\sum_{k,l,m\in X^n}
 g(l,m;[n/2])\, \left|z_k\, z_{k\oplus l\oplus m} - z_{k\oplus l}\, z_{k\oplus m}\right|^2
 \nonumber\\
 & &+   \sum_{k,l \in X^n} \left|z_k\, z_{k\oplus l}\right|^2 \sum_{m\in X^n} g(l\oplus m,m;[n/2]).
\end{eqnarray*}
Let us assume for a moment that
\begin{equation}
\label{eq:Y1}
\sum_{m\in X^n} g(l\oplus m,m;[n/2])=1, \qquad \forall l \in X^n.
\end{equation}
Then, the result follows by normalization (\ref{eq:normz}),  $z\in \mathbb{S}^{2N-1}$, since
$$\sum_{k,l \in X^n} \left|z_k\, z_{k\oplus l}\right|^2= \Big(\sum_{k \in X^n} |z_k|^2\Big)^2=1.$$
In fact, equality (\ref{eq:Y1}), is a consequence of the following lemma, for $n_A=[n/2]$.
\qed
\end{proof}
\begin{lemma}
The following equality holds
\begin{equation}
Y=\sum_{m\in X^n} g(l\oplus m,m; n_A )=1, \qquad \forall l \in X^n, \quad \forall n_A\in S.
\end{equation}
\end{lemma}
\begin{proof}
From Eq. (\ref{eq:gdef}) in Theorem \ref{th:Delta} we get
\begin{eqnarray*}
Y=\sum_{m\in X^n}\delta_0 \big((l\oplus m)\wedge m \big)\, \hat{g}(|l\oplus m|, |m|;n_A).
\end{eqnarray*}
Let us define the set $B=\{i\in S \,|\, l_i=0 \}\subset S$, so that $l_B=0$. We get
$(l\oplus m)\wedge m=0$ iff
$m_B=0$ and $(l_{\bar{B}} \oplus  m_{\bar{B}}) \wedge m_{\bar{B}}=0$. But the second equality is identically satisfied, because $l_{\bar{B}}$ is a vector of all 1. Thus
$\delta_0 \big((l\oplus m)\wedge m \big)=\delta_0(m_{\bar{B}})$ and we get
\begin{eqnarray*}
Y&=& \sum_{m\in X^n} \delta_0(m_{\bar{B}})\, \hat{g}(|l_{\bar{B}} \oplus  m_{\bar{B}}|, | m_{\bar{B}}|;n_A)\\
&=&
\sum_{m\in X^r} \hat{g}(|11\dots1 \oplus  m|, | m|;n_A)
\nonumber\\
&= & \sum_{m\in X^r} \hat{g}(r-|m|, |m|;n_A)
\nonumber\\
&= &\sum_{m\in X^r} \sum_{0\leq t\leq r}  \delta_t(|m|)\, \hat{g}(r-t, t ;n_A)
=\sum_{0\leq t\leq r}  \left(\begin{array}{c} r  \\ t \end{array} \right) \hat{g}(r-t, t ;n_A),
\end{eqnarray*}
where $r=|\bar{B}|=|l|.$ Let us now use the form  (\ref{eq:hatg1}) of the function $\hat{g}$ given in Remark \ref{th:Delta1},
\begin{eqnarray*}
Y&=&\frac{1}{2}\sum_{0\leq t\leq r}
\left(\begin{array}{c} r  \\ t \end{array} \right)
\left(\begin{array}{c} n \\ r-t, t\end{array}\right)^{-1}
\left[ \left(\begin{array}{c} {n_A} \\ r-t \end{array} \right)
\left(\begin{array}{c} {n_{\bar{A}}}  \\  t \end{array} \right)
+\left(\begin{array}{c} {n_A}    \\  t \end{array} \right)
\left(\begin{array}{c} {n_{\bar{A}}} \\ r-t \end{array} \right)\right]
\\
&=&\frac{1}{2}\left(\begin{array}{c} n \\ r\end{array}\right)^{-1}
\sum_{0\leq t\leq r}
\left[ \left(\begin{array}{c} {n_A} \\ r-t \end{array} \right)
\left(\begin{array}{c} {n_{\bar{A}}}  \\  t \end{array} \right)
+\left(\begin{array}{c} {n_A}    \\  t \end{array} \right)
\left(\begin{array}{c} {n_{\bar{A}}} \\ r-t \end{array} \right)\right]
\\
&=&\left(\begin{array}{c} n \\ r\end{array}\right)^{-1}
\sum_{0\leq t\leq r}
\left(\begin{array}{c} {n_A}    \\  t \end{array} \right)
\left(\begin{array}{c} {n_{\bar{A}}} \\ r-t \end{array} \right).
\end{eqnarray*}
By recalling  Vandermonde's identity \cite{comtet},
$$\sum_j \left(\begin{array}{c} m \\ j\end{array}\right)
\left(\begin{array}{c} n-m \\ k-j \end{array}\right)=
\left(\begin{array}{c} n \\ k\end{array}\right),$$
since $n_A+n_{\bar{A}}=n$, we get
\begin{eqnarray*}
Y&=&\left(\begin{array}{c} n \\ r\end{array}\right)^{-1}
\sum_{0\leq t\leq r}
\left(\begin{array}{c} {n_A}    \\  t \end{array} \right)
\left(\begin{array}{c} {n- n_A} \\ r-t \end{array} \right)=
\left(\begin{array}{c} n \\ r\end{array}\right)^{-1}
\left(\begin{array}{c} n \\ r\end{array}\right)=1.
\end{eqnarray*}
\qed
\end{proof}
\begin{remark}
\label{rm:avlinent}
Recall that, by Theorem \ref{th:fullyfactorized}, the potential of multipartite entanglement attains its upper bound $\pi_{\mathrm{ME}}(z)=1$ on fully factorized states. Thus in (\ref{eq:pime4}) the nonegative sum which is subtracted from unity represents the amount of entanglement of $\ket{\psi}$, and MMES  are those states that maximizes the distances $\left|z_k\, z_{k\oplus l\oplus m} - z_{k\oplus l}\, z_{k\oplus m}\right|^2$. In fact, the average over balanced bipartition of the linear entropy (\ref{eq:linent}) yields
\begin{equation}
\mathcal{L}_{\mathrm{ME}}(z)=\frac{N_A}{N_A-1}\left(1-\pi_{\mathrm{ME}}(z) \right),
\end{equation}
with $N_A= 2^{[N/2]}$. Thus, apart from a normalization factor, the sum in (\ref{eq:pime4})
is nothing but the average linear entropy. Note that the number of terms in the sum is $\mathcal{N}^{(4)}$ given in (\ref{eq:N4}), since the terms with $l=0$ or $m=0$ identically vanish.
\end{remark}
\begin{remark}
In the spirit of the above Remark, one can prove Theorem \ref{th:pime4} by following a different path. First, one can easily write an expression analogous to (\ref{th:fullyfactorized}) for the purity $\pi_A(z)$ of a given bipartition.
Incidentally, this would give an explicit expression of $\mathcal{L}_A(z)$. Then, one considers the average over balanced bipartitions and, by  noting that the proofs of Theorems \ref{th:Delta} and \ref{th:pme1} do not depend on the particular form of the monomials $z_k z_{l} \bar{z}_{m}\bar{z}_{n}$, that can be replaced by $|z_k z_{l}- z_{m}z_{n}|^2$, one obtains the desired result. By comparing the two proofs, since the average of 1 is 1, one can easily distillate an alternative combinatoric proof of Vandermonde's identity.
\end{remark}

\begin{example}
Consider $n=2$ qubits. One gets
\begin{eqnarray}
\pi_{\mathrm{ME}}(z)&=&
|z_{00}|^4 + |z_{01}|^4 + |z_{10}|^4 + |z_{11}|^4
\nonumber\\
& & + 2\left( |z_{00}|^2 |z_{01}|^2 + |z_{00}|^2 |z_{10}|^2 + |z_{11}|^2 |z_{01}|^2 + |z_{11}|^2 |z_{10}|^2\right)
\nonumber\\
& & + 4\Re(z_{00}\bar{z}_{01}\bar{z}_{10}z_{11})
\nonumber\\
&=& 1-2|z_{00}z_{11}-z_{01}z_{10}|^2.
\end{eqnarray}
The first equality follows from Corollary \ref{th:pme2}, while the second equality derives from Theorem \ref{th:pime4}. Note that the number of terms $\mathcal{N}^{(1)}=\mathcal{N}^{(2)}=4$ and $\mathcal{N}^{(4)}=1$ is in agreement with the counting of Theorem \ref{th:pme3e} and Remark \ref{rm:avlinent}. See Table \ref{tab:monomials}.
\end{example}

\begin{example}
For 3 qubits we will give the potential of multipartite entanglement in the form 4 of Theorem \ref{th:pime4}:
\begin{eqnarray}
\pi_{\mathrm{ME}}(z)&=&1 -2 \sum_{p\in\mathcal{C}_3} \Big( | z_{p(000)} z_{p(011)} - z_{p(001)} z_{p(010)} |^2
\nonumber\\
& & \phantom{1 -2 \sum_{p\in\mathcal{C}_3} \Big(}
+ | z_{p(100)} z_{p(111)} - z_{p(101)} z_{p(110)} |^2
\nonumber\\
& & \phantom{1 -2 \sum_{p\in\mathcal{C}_3} \Big(}
+\frac{1}{3}  | z_{p(100)} z_{p(011)} - z_{p(101)} z_{p(010)} |^2
\nonumber\\
& & \phantom{1 -2 \sum_{p\in\mathcal{C}_3} \Big(}
+\frac{1}{3}   | z_{p(000)} z_{p(111)} - z_{p(001)} z_{p(110)} |^2 \Big),
\end{eqnarray}
where the sum is over the $3$ cyclic permutations
\begin{equation}
\label{eq:C3def}
\mathcal{C}_3=\{s^i\, |\, s(1,2,3)=(2,3,1),\; i=0,1,2\}
\end{equation}
of the qubits $S=\{1,2,3\}$. Here $p(k)$ denotes the natural action of the permutation group on $k$,
\begin{equation}
\label{eq:Pnaction}
\mathcal{P}_n\times X^n \ni (p,k)  \mapsto p(k)=(k_{p(i)})_{i\in S} \in X^n.
\end{equation}
In agreement with Remark \ref{rm:avlinent}, the number of distinct terms is $\mathcal{N}^{(4)}=3 \times 4=12$.  See Table \ref{tab:monomials}.
\end{example}

Now we will focus on the problem of the existence of perfect MMES. In particular we will try to construct them by using characterization 2 of Theorem  \ref{prop:perfect}.
It is not obvious that a state with  $\pi_{\mathrm{ME}}=1/N_A$ exists: in order to find a solution one must solve for $\rho_{A}=1/N_A, \forall \; (A,\bar{A})$, and this set of equations might not admit a solution.

\section{Perfect MMES. Probabilistic approach}
\label{sec:perfect MMES}

We will look more closely at the equations
\begin{equation}
\rho_A=1 /N_A, \quad \text{for every subsystem } A\subset S \text{ with }  |A| \leq n/2,
\label{eq:rho1NA}
\end{equation}
that, according to Theorem \ref{prop:perfect}, characterize a perfect MMES. Although we could consider only maximal subsets $A\subset S$, with $|A|=[n/2]$, it will be more convenient to consider also smaller subsets $A$.

Let us first consider the diagonal elements in the computational basis $\{\ket{\ell}_A\}_{\ell \in X^{n_A}}\subset \mathcal{H}_A$. By Eq. (\ref{eq:rhoA}) of Theorem \ref{th:rhoA piA}, one gets
\begin{eqnarray}
\left\langle \ell |\rho_{A}|\ell \right\rangle &=&
\sum_{k , l\in X^n}
z_{k} \bar{z}_{l}
\delta_{k_{\bar{A}},l_{\bar{A}}} \delta_{k_A,\ell} \delta_{l_A,\ell}
= \sum_{k , l\in X^n}
z_{k} \bar{z}_{l}
\delta_{k,l} \delta_{k_A,\ell}\nonumber \\
&=& \sum_{k \in X^n}
|z_{k}|^2  \delta_{k_A,\ell} .
\end{eqnarray}
Therefore, from (\ref{eq:rho1NA}) we obtain
\begin{eqnarray}
\left\langle \ell |\rho_{A}| \ell\right\rangle = \sum_{k\in X^{n}}|z_{k}|^{2} \delta_{k_A, \ell}=1/N_A,
\label{eq:diagrho}
\end{eqnarray}
with $N_A=2^{|A|}$, $\forall \ell \in X^{|A|}$, $\forall A\subset S$, with $|A|\leq n/2$.

Now note that, due to normalization, $\sum |z_k|^2=1$, we can look at $(|z_k|^2)_{k\in X^n}$ as a probability vector on the finite space $X^n$ of $n$ \emph{classical}  bits. In view of this interpretation, we will introduce the
\begin{definition}[Population probability]
Given a state $\ket{\psi}\in\mathcal{H}_S$ and its Fourier coefficients $(z_k)_{k\in X^n}$ in the computational basis,  we define the \emph{population probability vector}
in the computational basis,
\begin{equation}
P_S(k) = |z_{k}|^{2},
\end{equation}
as the probability of the binary sequence $k=(k_i)_{i\in S}\in X^n$.
Moreover, let $E[\cdot]$  denote the expectation value with respect to $P_S$,
\begin{equation}
E[f(k)] = \sum_{k\in X^n} f(k) P_S(k) =  \sum_{k\in X^n} f(k) |z_{k}|^{2},
\end{equation}
for any function $f: X^n \to \mathbb{C}$.
\end{definition}
According to the above definition, Eq. (\ref{eq:diagrho}) reads
\begin{eqnarray}
E\left[ \delta_{k_A, \ell} \right] = 2^{-|A|},  \quad
\forall \ell \in X^{|A|}, \quad  \forall A\subset S, \; \text{with } |A|\leq n/2.
\end{eqnarray}
By noting that the above expectation value is nothing but the marginal probability
distribution
\begin{equation}
E\left[ \delta_{k_A, \ell} \right] = P_A (\ell),
\end{equation}
we have arrived at the following
\begin{theorem}
\label{th:flatmarg}
A necessary condition for a state $\ket{\psi}\in \mathcal{H}_S$ to be a perfect MMES is that all  the marginals over $n_A\leq n/2$ variables of its population probability vector in the computational basis $P_S(k) = |z_{k}|^{2}$, are completely random:
\begin{eqnarray}
P_{A}(\ell)= 2^{-|A|},  \quad
\forall \ell \in X^{|A|}, \quad  \forall A\subset S, \; \text{with } |A|\leq n/2.
\end{eqnarray} \qed
\end{theorem}
\begin{remark}
For $A=p(C)$ with $p\in \mathcal{P}_n$ and $C=\{1,2,\dots,n_A\}$ one can write
\begin{eqnarray}
E\left[
\prod_{j=1}^{n_A}\delta_{k_{p(j)},\ell_{j}}\right]= 2^{-n_A},  \quad
\forall \ell \in X^{n_A}, \;  \forall p\in \mathcal{P}_n, \quad \text{with } n_A\leq n/2,
\end{eqnarray}
which means that
\begin{eqnarray}
P_{p(C)}(\ell)= 2^{-n_A},  \quad
\forall \ell \in X^{n_A}, \;  \forall p\in \mathcal{P}_n, \quad \text{with } n_A\leq n/2.
\end{eqnarray}
\end{remark}
\begin{remark}
According to Theorem \ref{th:flatmarg}, a first step in the problem of seeking perfect MMES is the following: Search for all probability functions on $X^n$, whose marginals on $n_A\leq n/2$ variables are uniform.
\end{remark}
The solution to this problem is given by the following
\begin{theorem}[Perfect MMES population]
\label{th:populationMMES}
The population probability vector in the computational basis of a perfect MMES of $n$ qubits has the form
\begin{eqnarray}
|z_k|^2 = P_{S}(k)
= 2^{-n}+\sum_{\frac{n}{2}< r \leq n} \sum_{j\in [S^r]} c_{j}^{(r)}\prod_{1\leq l\leq r} (2 k_{j_l} -1), \quad k\in S,
\label{eq:PSk}
\end{eqnarray}
for some $c_{j}^{(r)}\in\mathbb{R}$, where $[S^r]=\{(j_1,\dots, j_r) \in S^r | j_1 < j_2<\dots <j_r\}$ denotes the set of ordered vectors of $S^r$.
\end{theorem}
\begin{proof}
Note that any function on $X^n$ is a multilinear function of the components of $k\in X^n$, because $k_i^2=k_i$.
Therefore, we can write
\begin{eqnarray*}
P_{S}(k)
= c^{(0)}+\sum_{r\in S} \sum_{j\in [S^r]} c_{j}^{(r)}\prod_{1\leq l\leq r} (2 k_{j_l} -1),
\end{eqnarray*}
which depends on the real parameters $c^{(r)}_j \in \mathbb{R}$, whose number is
$$1+\sum_{r\in S}|[S^r]|=\sum_{r\in S}
\left(\begin{array}{c} n \\ r \end{array} \right) = 2^n.$$
The normalization of $P_S$ implies that
$$
1=\sum_{k\in X^n} P_S(k)= 2^n   c^{(0)}+\sum_{r\in S} \sum_{j\in [S^r]} c_{j}^{(r)}\prod_{1\leq l\leq r} \sum_{k_{j_l}\in X} (2 k_{j_l} -1) = 2^n c^{(0)} ,
$$
that is
\begin{equation*}
c^{(0)}=2^{-n}.
\end{equation*}
Let us now consider a subset with one element $A=\{j\}$, with $j\in S$. For any $k_j\in X$ one must have
\begin{eqnarray*}
\frac{1}{2}=P_{\{j\}}(k_j)= \sum_{k\in X^{\bar{A}}} P_S(k)=2^{n-1}   c^{(0)}+c_{j}^{(1)}(2 k_{j} -1)= \frac{1}{2} + c_{j}^{(1)}(2 k_{j} -1),
\end{eqnarray*}
that is
$$c^{(1)}_j=0, \quad j\in S.$$
Analogously, for a subset with two elements $A=\{j_1,j_2\}$,
\begin{eqnarray*}
2^{-2}=P_A (k_{j_1},k_{j_2})= \sum_{k\in X^{\bar{A}}} P_S(k)=2^{-2}+c_{(j_1, j_2)}^{(2)}(2 k_{j_1} -1)(2 k_{j_2} -1),
\end{eqnarray*}
that is
$$c^{(2)}_j=0, \quad j\in [S^2].$$
By induction we get
\begin{eqnarray*}
c^{(r)}_j = 0, \quad \forall j \in  [S^r], \quad \text{for }  1\leq r \leq n/2,
\end{eqnarray*}
and the result follows.
\qed
\end{proof}
\begin{remark}
The range of the free parameters $c_{j}^{(r)}$ is determined by the inequalities
$0\leq P_S(k) \leq 1$, $\forall k\in S$. Their number is given by
\begin{eqnarray}
\sum_{\frac{n}{2}< r \leq n} | [S^r]| &=& \sum_{\frac{n}{2}< r \leq n}
\left( \begin{array}{c} n \\ r \end{array} \right) = \frac{1}{2} \sum_{\frac{n}{2}< r \leq n}
\left[\left( \begin{array}{c} n \\ r \end{array} \right) +
\left( \begin{array}{c} n \\ n-r \end{array} \right)
\right]\nonumber\\
&=&   \frac{1}{2} \sum_{0\leq r < \frac{n}{2}}
\left( \begin{array}{c} n \\ r \end{array} \right) + \frac{1}{2} \sum_{\frac{n}{2}< r \leq n}
\left( \begin{array}{c} n \\ r \end{array} \right)
\nonumber\\
&=& 2^{n-1}-\frac{1+(-1)^n}{4}\left(\begin{array}{c} n \\{[n/2]}\end{array}\right).
\label{eq:numcjr}
\end{eqnarray}
The particular solution $c_j^{(r)}=0$ for all $r\in S$  that yields a uniform probability $P_S(k)= 2^{-n}$ will play a role in the following.
\end{remark}

Theorem \ref{th:populationMMES}  completely determines the structure of the moduli $r_k=|z_k|=\sqrt{P_S(k)}$ of the Fourier coefficients $z_k$ of a perfect MMES in the computational basis. However this is only half of the work. In fact, the easy one. It remains to determine the phases, defined in the following
\begin{definition}
\label{def:modphas}
A state $\ket{\psi}\in\mathcal{H}_S$  can be expressed in the computational basis as
\begin{equation}
|\psi\rangle = \sum_{k\in X^n} z_k |k\rangle, \qquad z_k = r_k \zeta_k,
\end{equation}
where the \emph{Fourier moduli} belongs to the intersection of the positive hyperoctant with a hypersphere
\begin{equation}
r\in(\mathbb{R}_+)^N\cap \mathbb{S}^{N-1}=\{(r_k)_{k\in X^n} \,|\,  r_k\in \mathbb{R}_+, \sum_k r^2_k=1\}, \quad N=2^n,
\end{equation}
while the \emph{Fourier phases} belongs to the  torus $\mathbb{T}^N=(\mathbb{S}^1)^N$
\begin{equation}
\zeta\in\mathbb{T}^N=\{(\zeta_k)_{k\in X^n} \,|\,  \zeta_k\in \mathbb{C}, |\zeta_k|=1\}.
\end{equation}
\end{definition}
We will now show that the phases $\zeta$ of a perfect MMES are solutions to the system of the off-diagonal elements of the  equation $\rho_A=1/N_A$.
\begin{theorem}[Perfect MMES phases]
\label{th:phasesperfectMMES}
A state $\ket{\psi}\in\mathcal{H}_S$ is a perfect MMES iff its Fourier phases $\zeta$ in the computational basis  are  solutions to the equations
\begin{eqnarray}
& & \sum_{m \in X^{\bar{A}}}
r_{\ell\oplus m} r_{\ell'\oplus m} \zeta_{\ell\oplus m}\bar{\zeta}_{\ell'\oplus m} =0,
\nonumber \\
& &\forall \ell, \ell' \in X^{A}, \; \ell\neq\ell', \quad  \forall A\subset S, \; |A|=[n/2] ,
\label{eq:phases}
\end{eqnarray}
where $r_k=\sqrt{P_S(k)}$, with the population probability vector $P_S(k)$   given in Theorem \ref{th:populationMMES}, for some coefficients  $c_{j}^{(r)}\in\mathbb{R}$.
\end{theorem}
\begin{proof}
The off-diagonal elements of $\rho_A$ in Eq. (\ref{eq:rhoA}) of Theorem \ref{th:rhoA piA} read
\begin{eqnarray*}
\left\langle \ell |\rho_{A}|\ell' \right\rangle =
\sum_{k , l\in X^n}
z_{k} \bar{z}_{l}
\delta_{k_{\bar{A}},l_{\bar{A}}} \delta_{k_A,\ell} \delta_{l_A,\ell'}
= \sum_{k_{\bar{A}} \in X^{\bar{A}}}
z_{\ell\oplus k_{\bar{A}}} \bar{z}_ {\ell'\oplus k_{\bar{A}}},
\end{eqnarray*}
$\forall \ell, \ell' \in X^{A}$,  $\ell\neq\ell'$.
Thus, by Equation (\ref{eq:rho1NA}), Definition \ref{def:modphas} and Theorem  \ref{th:populationMMES} one gets the desired result.
\qed
\end{proof}
\begin{remark}
An alternative form of (\ref{eq:phases}) in terms of the permutation group is
\begin{eqnarray}
\sum_{k, l \in X^n} \sqrt{P_S(k) P_S(l)} \zeta_k\bar{\zeta}_l
\prod_{\frac{n}{2}<j \leq n}  \delta_{k_{p(j)},l_{p(j)}}
\prod_{1\leq j\leq \frac{n}{2}}\delta_{k_{p(j)},\ell_j} \prod_{1\leq j\leq \frac{n}{2}}\delta_{l_{p(j)},\ell'_j}=0, \nonumber \\
\forall p\in {\mathcal{P}}_n,  \quad \forall\ell,\ell' \in X^{n_A}, \quad  \ell\neq\ell' .
\quad 1\leq n_{A}\leq n/2. \nonumber\\
\end{eqnarray}
\end{remark}
Let us now  investigate whether the system of equations (\ref{eq:phases}) admits a solution or not. In particular, it is important to count the number of equations and of variables and to look for which values of $n$ the system is over-determined.
\begin{theorem}[Number of equations and variables]
\label{th:eqs vs vars}
The set of equations (\ref{eq:phases}) determining a perfect MMES is a system of
\begin{equation}
m_e=2^{[n/2]}(2^{[n/2]}-1) \left(\begin{array}{c}n \\ {[n/2]} \end{array}\right)
\label{eq:me}
\end{equation}
real equations involving
\begin{equation}
m_x =
3\cdot 2^{n-1}-\frac{1+(-1)^n}{4}\left(
\begin{array}{c} n \\ {[n/2]} \end{array} \right)
\label{eq:mx}
\end{equation}
real variables.
\end{theorem}
\begin{proof}
By noting that exchanging $\ell$ and $\ell'$ one obtains the complex conjugate, the counting of real equations coincides with the total counting of equations  (\ref{eq:phases}). Since $\ell\neq\ell'$, we get
$$m_e=|X^{A}|\left(|X^{A}|-1\right)\, \#(A),$$ where $\#(A)$ is the number of maximal subsets $A\subset S$. Now, $|X^{A}|=2^{[n/2]}$ and $\#(A)=\left(\begin{array}{c}n \\ {[n/2]} \end{array}\right)$, and (\ref{eq:me}) follows.
On the other hand, the variables are the $2^n$ phases $\zeta$ and the parameters $c^{(r)}_j$, whose number is given by (\ref{eq:numcjr}), for a total number of $m_x$ variables.
\end{proof}
\begin{remark}
For large values of $n$, by  Stirling's approximation one gets
\begin{equation}
m_e \sim 2^{2 n} \sqrt{\frac{2}{n \pi}} , \qquad m_x\sim 3 \; 2^{n-1}, \qquad n\to\infty.
\end{equation}
\begin{table}
\caption{Number of equations vs number of variables}
\label{tab:e vs x}       
\centering
\begin{tabular}{lll}
\hline\noalign{\smallskip}
$n$ & $m_e$ & $m_x$  \\
\noalign{\smallskip}\hline\noalign{\smallskip}
2 & 4 & 5 \\
3 & 6 & 12 \\
4 & 72 & 21 \\
5 & 120 & 48\\
6 & 1120 & 86\\
7 & 1960 & 192\\
8 & 16800 & 349\\
$n\to\infty$ & $\sqrt{2/n \pi}\, 2^{2 n}$  & $3\; 2^{n-1}$\\
\noalign{\smallskip}\hline
\end{tabular}
\end{table}
As shown in Table~\ref{tab:e vs x}, for $n\geq 4$ the number of equations is larger than the number of variables and the system is overdetermined. Therefore, symmetries must play a crucial role in order to assure the existence of a solution.
\end{remark}

\subsection{Examples}

\subsubsection{Two qubits}
Let us consider the case of $n=2$ qubits.  $S=\{1,2\}$ and we get
from (\ref{eq:PSk})
\begin{equation}
r_k^2=P_{\{1,2\}}(k)=\frac{1}{4}(1+c \sigma_1\sigma_2),
\end{equation}
with  $\sigma_i=(2 k_i-1)$, $i\in S$. Normalization and positivity, $0\leq P_{\{1,2\}}(k)\leq 1$, $\forall k\in X^2$ imply that  $c\in [-1,1]$.
Equation (\ref{eq:phases}) particularizes to
\begin{eqnarray}
\begin{cases}
r_{00}r_{10}\zeta_{00}\bar{\zeta}_{10}+r_{01}r_{11} \zeta_{01}\bar{\zeta}_{11}=0 \\
r_{00}r_{01}\zeta_{00}\bar{\zeta}_{01}+r_{10}r_{11}\zeta_{10}\bar{\zeta}_{11}=0,
\end{cases}
\end{eqnarray}
and, by noting that  $r^2_{00}=r^2_{11}=(1+c)/4$ and
$r^2_{01}=r^2_{10}=(1-c)/4$, one gets
\begin{eqnarray}
\begin{cases}
\sqrt{1-c^{2}}\left(\zeta_{00}\bar{\zeta}_{10}+\zeta_{01}\bar{\zeta}_{11}\right)=0\\
\sqrt{1-c^{2}}\left(\zeta_{00}\bar{\zeta}_{01}+\zeta_{10}\bar{\zeta}_{11}\right)=0
\end{cases}
.
\end{eqnarray}
The above system reduces to a single equation
\begin{equation}
\sqrt{1-c^{2}}\left(\zeta_{00}\zeta_{11}\bar{\zeta}_{10}\bar{\zeta}_{01}+1\right)=0 .
\end{equation}
\begin{enumerate}
\item
A first class of solutions is $|c|=1$ and arbitrary phases. This yields,
for $c=1$, $r_{01}=r_{10}=0$ and  $r_{00}=r_{11}=1/\sqrt{2}$, whence
\begin{equation}
|\psi\rangle=\frac{1}{\sqrt{2}}\left(\zeta_{00}|00\rangle+\zeta_{11}|11\rangle\right),
\end{equation}
while, for $c=-1$, $r_{00}=r_{11}=0$ and $r_{01}=r_{10}=1/\sqrt{2}$ whence
\begin{equation}
|\psi\rangle=\frac{1}{\sqrt{2}}\left(\zeta_{01}|01\rangle+\zeta_{10}|10\rangle\right).
\end{equation}
The above states are known as Bell states. They are, obviously, maximally bipartite entangled. Indeed, for $n=2$ multipartite entanglement reduces to bipartite entanglement.

\item On the other hand, when $|c|\neq 1$, the  perfect MMES are
\begin{equation}
|\psi\rangle=\frac{1}{2}\sqrt{1+c}\left(\zeta_{00}|00\rangle+\zeta_{11}|11\rangle\right)+\frac{1}{2}\sqrt{1-c}\left(\zeta_{01}|01\rangle+\zeta_{10}|10\rangle\right),
\end{equation}
where the phases must satisfy the condition
\begin{equation}
\zeta_{00}\zeta_{11}=-\zeta_{01}\zeta_{10}.
\end{equation}
Therefore,
\begin{equation}
|\psi\rangle=\alpha \frac{\sqrt{1+c}}{2}\left(\beta|00\rangle+\bar{\beta}|11\rangle\right)- i \alpha \frac{\sqrt{1-c}}{2}\left(\gamma|01\rangle+\bar{\gamma}|10\rangle\right),
\end{equation}
with $\alpha=( \zeta_{00}\zeta_{11})^{1/2}$, $\beta=( \zeta_{00}\bar{\zeta}_{11})^{1/2}$ and $\gamma=(\zeta_{01}\bar{\zeta}_{10})^{1/2}$.

\item The particular  case $c=0$ corresponds to a uniform amplitude distribution $r_k=1/2$, $k\in X^2$. To such a class belong perfect MMES with phases that are $\pm 1$.
\begin{equation}
\ket{\psi}=\frac{1}{2} \sum_{k\in X^2} \zeta_k \ket{k}, \quad \zeta_{k}\in\{-1,+1\}, \quad   \prod_{k\in X^2}\zeta_{k}=-1.
\label{eq:2real}
\end{equation}

\end{enumerate}

\subsubsection{Three  qubits}
Let us consider the case of $n=3$ qubits.  $S=\{1,2,3\}$ and we get from (\ref{eq:PSk})
\begin{equation}
r_{k}^{2}= P_{1,2,3}(k)=\frac{1}{8}\left(1+c_{1}\sigma_{2}\sigma_{3}+c_{2}\sigma_{1}\sigma_{3}+c_{3}\sigma_{1}\sigma_{2}+d\sigma_{1}\sigma_{2}\sigma_{3}\right),
\label{eq:rk3}
\end{equation}
where $\sigma_i=(2 k_i -1)$, with $i\in S$.
On the other hand, from (\ref{eq:phases})  we obtain
\begin{equation}
\begin{cases}
z_{000}\bar{z}_{100}+z_{001}\bar{z}_{101}+z_{010}\bar{z}_{110} +z_{011}\bar{z}_{111} =0 \\
z_{000}\bar{z}_{010} + z_{001}\bar{z}_{011}+z_{100}\bar{z}_{110}+z_{101}\bar{z}_{111}=0 \\
z_{000}\bar{z}_{001}+z_{010}\bar{z}_{011}+z_{100}\bar{z}_{101}+z_{110}\bar{z}_{111}=0
\end{cases}
.
\end{equation}
Note that the three equations are obtained by a cyclic permutation of the three qubits $S$.

\begin{enumerate}

\item
If $c_i=0$  ($i\in S$) and $d\in[-1,1]$,  one gets from (\ref{eq:rk3})
\begin{equation}
 r^2_k= \frac{1}{8}\left(1+d\sigma_{1}\sigma_{2}\sigma_{3}\right)= \frac{1-(-1)^{|k|} |d|}{8}, \qquad k\in X^3.
\end{equation}
Thus,
\begin{equation} \label{sist3qubit}
\begin{cases}
\sqrt{1-d^{2}}\left(\zeta_{000}\bar{\zeta}_{100}+\zeta_{001}\bar{\zeta}_{101}+\zeta_{010}\bar{\zeta}_{110}+\zeta_{011}\bar{\zeta}_{111}\right)=0 \\
\sqrt{1-d^{2}}\left(\zeta_{000}\bar{\zeta}_{010}+\zeta_{001}\bar{\zeta}_{011}+\zeta_{100}\bar{\zeta}_{110}+\zeta_{101}\bar{\zeta}_{111}\right)=0 \\
\sqrt{1-d^{2}}\left(\zeta_{000}\bar{\zeta}_{001}+\zeta_{010}\bar{\zeta}_{011}+\zeta_{100}\bar{\zeta}_{101}+\zeta_{110}\bar{\zeta}_{111}\right)=0
\end{cases}
.
\end{equation}

\begin{enumerate}
\item
If  $|d|=1$ the phases are arbitrary and the MMES is
\begin{eqnarray}
\ket{\psi}&=&\frac{1}{2}\left(\zeta_{001}|001\rangle+\zeta_{010} |010\rangle+\zeta_{100} |100\rangle+\zeta_{111} |111\rangle\right)\nonumber\\
&=&\frac{1}{2}\sum_{|k|\, \mathrm{odd}} \zeta_k \ket{k},
\label{meps3.1}
\end{eqnarray}
when $d=1$, and
\begin{eqnarray} \label{meps3.2}
\ket{\psi}&=&\frac{1}{2}\left(\zeta_{000} |000\rangle+\zeta_{011}|011\rangle+\zeta_{101}|101\rangle+\zeta_{110}|110\rangle\right)
\nonumber\\
&=&\frac{1}{2}\sum_{|k|\, \mathrm{even}} \zeta_k \ket{k},
\end{eqnarray}
when $d=-1$

\item
When $|d|\neq 1$, the phases must satisfy
\begin{equation}
\begin{cases}
\zeta_{000}\bar{\zeta}_{100}+\zeta_{001}\bar{\zeta}_{101}=\alpha \\
\zeta_{000}\bar{\zeta}_{010}+\zeta_{001}\bar{\zeta}_{011}=\beta \\
\zeta_{000}\bar{\zeta}_{001}+\zeta_{010}\bar{\zeta}_{011}=\gamma \\
\zeta_{010}\bar{\zeta}_{110}+\zeta_{011}\bar{\zeta}_{111}=-\alpha \\
\zeta_{100}\bar{\zeta}_{110}+\zeta_{101}\bar{\zeta}_{111}=-\beta \\
\zeta_{100}\bar{\zeta}_{101}+\zeta_{110}\bar{\zeta}_{111}=-\gamma
\end{cases},
\end{equation}
with $|\alpha|,|\beta|,|\gamma|\leq 2$.
It is a system of $6$ equations in $8$ variables. Thus the general solutions, for fixed $d$, live on a $5$-dimensional manifold. A particular $3$-dimensional submanifold is obtained by $\alpha=\beta=\gamma=0$. In such a case
\begin{equation}
\begin{cases}
\zeta_{000}\bar{\zeta}_{100}=- \zeta_{001}\bar{\zeta}_{101} \\
\zeta_{010}\bar{\zeta}_{110}= -\zeta_{011}\bar{\zeta}_{111} \\
\zeta_{000}\bar{\zeta}_{010}= -\zeta_{001}\bar{\zeta}_{011} \\
\zeta_{100}\bar{\zeta}_{110}= -\zeta_{101}\bar{\zeta}_{111} \\
\zeta_{000}\bar{\zeta}_{001}= -\zeta_{010}\bar{\zeta}_{011} \\
\zeta_{100}\bar{\zeta}_{101}= -\zeta_{110}\bar{\zeta}_{111}
\end{cases}
.
\end{equation}
For example, the following MMES is an element of that manifold when $d=0$
\begin{eqnarray}
|\psi\rangle&=&\frac{1}{\sqrt{8}}\big(-|000\rangle+|001\rangle+|010\rangle+|011\rangle
\nonumber\\
& & \phantom{\frac{1}{\sqrt{8}}\big(}
+|100\rangle+|101\rangle+|110\rangle-|111\rangle\big).
\end{eqnarray}
As in the case of 2 qubits, this is an example of perfect MMES with uniform amplitudes $r_k=1/\sqrt{8}$ and real phases $\zeta_k\in\{-1,1\}$, with $k\in X^3$.
\end{enumerate}

\item
If $d=0$ and $c_i=c$ ($i\in S$) with $c\in[-1/3,1]$, one gets
\begin{equation}
 r^2_k= \frac{1}{8}\left(1+c(\sigma_{1}\sigma_{2}+\sigma_2\sigma_{3}+\sigma_3\sigma_1)\right)= \begin{cases}
\frac{1}{8}(1+3c) & \text{for } k\in\{000, 111\} \\
\frac{1}{8}(1-c) & \text{otherwise}. \\
\end{cases}
\end{equation}
Thus,
\begin{equation} \label{sist3qubit1}
\begin{cases}
\sqrt{(1+3c)(1-c)}\left(\zeta_{000}\bar{\zeta}_{100}+\zeta_{011}\bar{\zeta}_{111}\right) \\
\qquad\qquad\qquad\qquad\qquad\qquad + (1-c)\left(\zeta_{001}\bar{\zeta}_{101}+\zeta_{010}\bar{\zeta}_{110}\right)=0 \\
\sqrt{(1+3c)(1-c)}\left(\zeta_{000}\bar{\zeta}_{010}+\zeta_{101}\bar{\zeta}_{111}\right)\\
\qquad\qquad\qquad\qquad\qquad\qquad
+(1-c)\left(\zeta_{001}\bar{\zeta}_{011}+\zeta_{100}\bar{\zeta}_{110}\right)=0 \\
\sqrt{(1+3c)(1-c)}\left(\zeta_{000}\bar{\zeta}_{001}+\zeta_{110}\bar{\zeta}_{111}\right)\\
\qquad\qquad\qquad\qquad\qquad\qquad
+ (1-c)\left(\zeta_{010}\bar{\zeta}_{011}+\zeta_{100}\bar{\zeta}_{101}\right)=0
\end{cases}
.
\end{equation}

\begin{enumerate}
\item
If  $c=1$ the phases are arbitrary and the perfect MMES is
\begin{equation}
\ket{\psi}=\frac{1}{\sqrt{2}}\left(\zeta_{000}|000\rangle+\zeta_{111} |111\rangle\right),
\end{equation}
As a particular case, when $\zeta_{000}=\zeta_{111}=1$, one obtains the GHZ state \cite{ghz}.

\item
For $c<1$, the solutions live on a $5$-dimensional submanifold.
Note that if one tries a solution for which the phases are independent of $c$ one gets
\begin{equation}
\begin{cases}
\zeta_{000}\bar{\zeta}_{100}+\zeta_{011}\bar{\zeta}_{111}=0 \\
\zeta_{000}\bar{\zeta}_{010}+\zeta_{101}\bar{\zeta}_{111}=0 \\
\zeta_{000}\bar{\zeta}_{001}+\zeta_{110}\bar{\zeta}_{111}=0\\
\zeta_{001}\bar{\zeta}_{101}+\zeta_{010}\bar{\zeta}_{110}=0 \\
\zeta_{001}\bar{\zeta}_{011}+\zeta_{100}\bar{\zeta}_{110}=0 \\
\zeta_{010}\bar{\zeta}_{011}+\zeta_{100}\bar{\zeta}_{101}=0
\end{cases}
,
\end{equation}
that is,
\begin{equation}
\begin{cases}
\zeta_{100}\zeta_{011}=\zeta_{010}\zeta_{101}= \zeta_{001}\zeta_{110}=-\zeta_{000}\zeta_{111}\\
\zeta_{001}\zeta_{110}+\zeta_{010}\zeta_{101}=0 \\
\zeta_{001}\zeta_{110}+\zeta_{100}\zeta_{011}=0 \\
\zeta_{010}\zeta_{101}+\zeta_{100}\zeta_{011}=0
\end{cases}
,
\end{equation}
which has no solutions.

\end{enumerate}

\end{enumerate}

\section{Uniform MMES}
\label{sec:uniform MMES}

According to Theorem \ref{th:populationMMES}, a perfect MMES has  a population probability vector in the computational basis given by (\ref{eq:PSk}), whose marginals on maximal subsystems are all uniform. In particular, a uniform probability vector is compatible with a perfect MMES.
In this Section we will focus just on this  class of states, that have uniform amplitudes
\begin{equation}
r_k=|z_k|= \sqrt{P_S(k)}= 1/\sqrt{N}, \qquad \forall k\in X^n,
\end{equation}
and depend only on $N=2^n$ phases.
\begin{definition}[Uniform states]
A state $\ket{\psi}\in \mathcal{H}_S$ of the form
\begin{equation}
|\psi\rangle = \frac{1}{\sqrt{N}}\sum_{k\in X^n} \zeta_k|k\rangle, \qquad  \zeta=(\zeta_k)\in \mathbb{T}^N,\qquad  N=2^n ,
\label{eq:randomphase}
\end{equation}
is said to have \emph{uniform amplitudes} in the computational basis. A state with uniform amplitudes in the computational basis is also called a \emph{uniform state}.
\end{definition}

First of all, we have a complete characterization of uniform maximizers of the potential of multipartite entanglement.
\begin{theorem}[Fully factorized uniform states]
The fully factorized states with uniform amplitudes, $z=\zeta/\sqrt{N}$,  have $\zeta_k=\prod_{i\in S} \zeta_{k_i}^i$, with   $\zeta_{k_i}^i \in \mathbb{S}^1$,  $k\in X^n$.
\end{theorem}
\begin{proof}
The result is an immediate consequence of Theorem \ref{th:fullyfactorized}, by observing that $z_k=\prod_{i\in S} |\alpha^i_{k_i}| \zeta^i_{k_i}=\prod_{i\in S} |\alpha^i_{k_i}| \prod_{j\in S} \zeta^j_{k_j}= \zeta_k/\sqrt{N}$.
\qed
\end{proof}

The various expressions of purity of a bipartition $(A,\bar{A})$ considered in Section  \ref{sec:bipartite}  simplify for uniform states. In particular, by plugging (\ref{eq:randomphase}) into
(\ref{eq:splitA}) we find
\begin{theorem}[Purity for uniform states]
\label{th:piAuniform}
Consider a state with uniform amplitudes in the computational basis $z=\zeta/\sqrt{N}$, with $\zeta\in\mathbb{T}^N$. Then for any bipartition $(A,\bar{A})$,
\begin{equation}
\pi_{A}(\zeta)=\frac{N_A+N_{\bar{A}}-1}{N}
+\frac{1}{N^2}\sum_{k\in X^S}\sum_{l\in X^A_*} \sum_{m\in X^{\bar{A}}_*}
\Re\left(
\zeta_{k}\,\bar{\zeta}_{k\oplus l}\,\zeta_{k\oplus l\oplus m}\,\bar{\zeta}_{k\oplus m}\right),
\label{eq:piAuniform}
\end{equation}
where $N_A=2^{n_A}$, and $N=2^n$.
\end{theorem}
\begin{proof}
When $z=\zeta/\sqrt{N}$, in the first three sums  of  (\ref{eq:splitA}) all terms are equal to $1/N^2$. Their number, according to Remark \ref{rm:splitA}, is $\mathcal{N}_{\mathrm{tot}}^{(1)}+\mathcal{N}_{\mathrm{tot}}^{(2)}= N(N_A + N_{\bar{A}} -1)$, and the result follows.
\qed
\end{proof}
\begin{remark}
Note that the first term on the
right-hand side corresponds to the
average entanglement  for typical states \cite{entrandom,aaa5,LLoyd,Lubkin,Page}, whose phases are uniformly distributed on the torus $\mathbb{T}^N$. Thus, the
combination of phases in the second term can increase or reduce the
value of the purity with respect to the typical one (at a fixed
bipartition).

Finally observe that, by setting $\zeta_k=\mathrm{e}^{i \varphi_k}$, with $\varphi_k\in [0,2\pi)$, $k\in X^n$, one gets
\begin{eqnarray}
\pi_{A}(\zeta)&=&\frac{N_A+N_{\bar{A}}-1}{N}\nonumber \\
&+& \frac{1}{N^2}\sum_{k\in X^S}\sum_{l\in X^A_*} \sum_{m\in X^{\bar{A}}_*}
\cos\left(
\varphi_{k}-\varphi_{k\oplus l}+\varphi_{k\oplus l\oplus m}-\varphi_{k\oplus m}\right).
\end{eqnarray}
\end{remark}
\begin{remark}
Note that for a uniform fully factorized state, since
$\zeta_{k}=\zeta^A_{k_A} \zeta^{\bar{A}}_{k_{\bar{A}}}$, with $\zeta^A_l=\prod_{i\in A} \zeta^{i}_{l_i}$, $\forall k\in X^n$,
one gets
\begin{equation}
\zeta_{k}\,\bar{\zeta}_{k\oplus l}\,\zeta_{k\oplus l\oplus m}\,\bar{\zeta}_{k\oplus m}=
(\zeta_{k_A}^A \zeta_{k_{\bar{A}}}^{\bar{A}}) (\bar{\zeta}_{k_A\oplus l}^A \bar{\zeta}_{k_{\bar{A}}}^{\bar{A}}) (\zeta_{k_A\oplus l}^A \zeta_{k_{\bar{A}}\oplus m}^{\bar{A}}) (\bar{\zeta}_{k_A}^A \bar{\zeta}_{k_{\bar{A}}\oplus m}^{\bar{A}})=1,
\end{equation}
$\forall k \in X^S$, $\forall l\in X^A$ , $\forall l\in X^{\bar{A}}$. Therefore, all terms of the sum
in (\ref{eq:piAuniform}) are 1, and
\begin{equation}
\pi_{A}(\zeta)=\frac{N_A+N_{\bar{A}}-1}{N}+ \frac{1}{N^2} N (N_A-1)(N_{\bar{A}}-1)=1,
\end{equation}
as it should.
\end{remark}
The counterpart of Theorem \ref{th:piAuniform} for the potential of multipartite entanglement is stated in the following
\begin{theorem}[Potential for uniform states]
\label{th:pimeunif}
If the Fourier amplitudes in the computational basis  are uniform, $z=\zeta/\sqrt{N}$, with $\zeta\in\mathbb{T}^N$, then the potential of multipartite entanglement reads
\begin{eqnarray}
\pi_{\mathrm{ME}}(\zeta)&=& \frac{N_A+N_{\bar{A}}-1}{N}  \nonumber\\
&+&\frac{1}{N^2}\sum_{k\in X^n} \sum_{l,m\in X^n_*}
 g(l,m;[n/2])\, \Re\left(\zeta_k\, \zeta_{k\oplus l\oplus m}\, \bar{\zeta}_{k\oplus l}\,\bar{\zeta}_{k\oplus m}\right),\nonumber\\
 \label{eq:pimeunif}
\end{eqnarray}
where $N_A=2^{[n/2]}$, $N_{\bar{A}}=2^{[(n+1)/2]}$, and $N=2^n$.
\end{theorem}
\begin{proof}
When $z=\zeta/\sqrt{N}$, in the first two sums  of  (\ref {eq:splitpme}) all terms are equal to $1/N^2$ and one obtains
\begin{eqnarray*}
\pi_{\mathrm{ME}}(\zeta)&=& \frac{1}{N} + \frac{2}{N} \sum_{l\in X^n_*} \hat{g}(|l|,0;[n/2])
\nonumber\\
&+&\frac{1}{N^2}\sum_{k\in X^n} \sum_{l,m\in X^n_*}
 g(l,m;[n/2])\, \Re\left(\zeta_k\, \zeta_{k\oplus l\oplus m}\, \bar{\zeta}_{k\oplus l}\,\bar{\zeta}_{k\oplus m}\right).
 \end{eqnarray*}
 We get
\begin{equation*}
\sum_{l\in X^n_*} \hat{g}(|l|,0;n_A)=\sum_{l\in X^n_*} \sum_{0\leq s\leq n}  \delta_s(|l|)\, \hat{g}(s, 0 ;n_A)
=\sum_{1\leq s \leq n}  \left(\begin{array}{c} n  \\ s \end{array} \right) \hat{g}(s, 0 ;n_A)
\end{equation*}
and, from (\ref{eq:hatg1}),
\begin{equation*}
\hat{g}(s, 0 ;n_A)= \frac{1}{2}\left(\begin{array}{c} n \\ s \end{array}\right)^{-1}  \left[
\left(\begin{array}{c} {n_A}  \\ s \end{array} \right)
+\left(\begin{array}{c} {n_{\bar{A}}} \\ s \end{array} \right)\right] .
\end{equation*}
Thus,
\begin{equation*}
2 \sum_{l\in X^n_*} \hat{g}(|l|,0;n_A)=\sum_{1\leq s \leq n}
 \left[
\left(\begin{array}{c} {n_A}  \\ s \end{array} \right)
+\left(\begin{array}{c} {n_{\bar{A}}} \\ s \end{array} \right)\right]=2^{n_A} + 2^{n_{\bar{A}}}-2,
\end{equation*}
and, by setting $n_A=[n/2]$, the result follows.
\qed
\end{proof}
We will now use Theorem \ref{th:pimeunif} and look for the uniform minimizers of the potential of multipartite entanglement.

\subsection{Two qubits}
For two qubits $n=2$, we have $N=4$, $N_A=N_{\bar{A}}=2$ and (\ref{eq:pimeunif}) becomes
\begin{eqnarray}
\pi_{\mathrm{ME}}(\zeta)&=&\frac{3}{4}+\frac{1}{16}\sum_{k\in X^2} \sum_{l,m\in X^2_*}
 g(l,m;1)\, \Re\left(\zeta_k\, \zeta_{k\oplus l\oplus m}\, \bar{\zeta}_{k\oplus l}\,\bar{\zeta}_{k\oplus m}\right)\nonumber\\
&=&\frac{3}{4}+\frac{1}{8}\sum_{k\in X^2}
 \hat{g}(1,1;1)\, \Re\left(\zeta_k\, \zeta_{k\oplus 11}\, \bar{\zeta}_{k\oplus 01}\,\bar{\zeta}_{k\oplus 10}\right).
\end{eqnarray}
From (\ref{eq:hatg}) we get  $\hat{g}(1,1;1)$, hence
\begin{equation}
\pi_{\mathrm{ME}}(\zeta)=\frac{3}{4}+\frac{1}{4}\Re\left(
\zeta_{00}\zeta_{11}\bar{\zeta}_{01}\bar{\zeta}_{10}\right). \label{eq:costfunction2}
\end{equation}
Uniform perfect MMES are solutions of the equation
\begin{equation}
\pi_{\mathrm{ME}}(\zeta)=\frac{1}{2},  \qquad \zeta\in \mathbb{T}^4,
\end{equation}
that is
\begin{equation}
\zeta_{00}\zeta_{11}\bar{\zeta}_{01}\bar{\zeta}_{10}=-1,
\end{equation}
which yields
\begin{equation}
\ket {\psi_2}=\frac{1}{2}\left(\zeta_{00}\ket{00}
+\zeta_{01} \ket{01}+\zeta_{10}\ket{10}
-\bar{\zeta}_{00}\zeta_{01}\zeta_{10}\ket{11}\right).
\label{eq:optimal2}
\end{equation}
In this degenerate case, multipartite entanglement coincides with
bipartite entanglement, and this state is obviously equivalent, up to local
unitaries, to a Bell state. A particular subclass is formed by uniform perfect MMES (\ref{eq:optimal2}) with real phases $\zeta\in\{-1,+1\}^4$. Their number is $2^3$ and has been already found by using a probabilistic approach. See (\ref{eq:2real}).

\subsection{Three qubits}

For $n=3$ qubits, $N=8$, $N_A=2$, $N_{\bar{A}}=4$, and  one must look for the solutions of
\begin{equation}
\pi_{\mathrm{ME}}(\zeta)=\frac{1}{2}, \qquad \zeta\in \mathbb{T}^8,
\end{equation}
where, from (\ref{eq:pimeunif}),
\begin{equation}
\pi_{\mathrm{ME}}(\zeta)=\frac{5}{8}+\frac{1}{64}\sum_{k\in X^3} \sum_{l,m\in X^3_*}
 g(l,m;1)\, \Re\left(\zeta_k\, \zeta_{k\oplus l\oplus m}\, \bar{\zeta}_{k\oplus l}\,\bar{\zeta}_{k\oplus m}\right)
\end{equation}
Due to the constraint $\delta_0(l \wedge m)$ in the coupling function $g$, see Theorem \ref{th:Delta}, one can  easily see that the pairs that yield nonvanishing contributions to the sum are
\begin{equation}
(l,m)=\big(p(001),p(010)\big), \qquad (l,m)=\big(p(001),p(110)\big),  \quad p\in\mathcal{C}_3,
\end{equation}
and the pairs obtained by exchanging $l$ and $m$, where $\mathcal{C}_3\subset\mathcal{P}_3$ is the subgroup of the $3$ cyclic permutations defined in (\ref{eq:C3def}).
Therefore,
\begin{eqnarray}
\pi_{\mathrm{ME}}(\zeta)&=&\frac{5}{8}+\frac{1}{32}\sum_{k\in X^3} \sum_{p\in\mathcal{C}_3} \Big[
 \hat{g}(1,1;1)\, \Re\left(\zeta_k\, \zeta_{k\oplus p(011)}\, \bar{\zeta}_{k\oplus p(001)}\,\bar{\zeta}_{k\oplus p(010}\right)\nonumber\\
& &\phantom{\frac{5}{8}+\frac{1}{32}\sum_{k\in X^3} \sum_{p\in\mathcal{C}_3}  }
 + \hat{g}(1,2;1)\, \Re\left(\zeta_k\, \zeta_{k\oplus p(111)}\, \bar{\zeta}_{k\oplus p(001)}\,\bar{\zeta}_{k\oplus p(110}\right) \Big] \nonumber\\
&=&\frac{5}{8}+\frac{1}{192} \sum_{p\in\mathcal{C}_3} \sum_{k\in X^3} \Big[
2 \Re\left(\zeta_{p(k)}\, \zeta_{p(k\oplus 011)}\, \bar{\zeta}_{p( k\oplus 001)}\,\bar{\zeta}_{p( k\oplus 010}\right)\nonumber\\
& &\phantom{\frac{5}{8}+\frac{1}{192}\sum_{k\in X^3} \sum_{p\in\mathcal{C}_3}  }
 + \Re\left(\zeta_{p(k)}\, \zeta_{p( k\oplus 111)}\, \bar{\zeta}_{p( k\oplus 001)}\,\bar{\zeta}_{p( k\oplus 110}\right) \Big] ,
\end{eqnarray}
since from (\ref{eq:hatg}) we get  $\hat{g}(1,1;1)=1/3$ and $\hat{g}(1,2;1)=1/6$.
By performing the sum over $X^3$ we finally obtain
\begin{eqnarray}
\pi_{\mathrm{ME}}(\zeta)&=&\frac{5}{8} +  \frac{1}{48}\sum_{p\in\mathcal{C}_3}\Big[
2 \Re(\zeta_{p(000)} \zeta_{p(011)}\bar{\zeta}_{p(001)}\bar{\zeta}_{p(010)})\nonumber\\
& &\phantom{\frac{5}{8} +  \frac{1}{48}\sum_{p\in\mathcal{C}_3}}
+  2 \Re(\zeta_{p(111)} \zeta_{p(100)} \bar{\zeta}_{p(110)}\bar{\zeta}_{p(101)})
\nonumber\\
& &\phantom{\frac{5}{8} +  \frac{1}{48}\sum_{p\in\mathcal{C}_3}}
+\Re(\zeta_{p(000)}\zeta_{p(111)}\bar{\zeta}_{p(001)}\bar{\zeta}_{p(110)})
\nonumber\\
& &\phantom{\frac{5}{8} +  \frac{1}{48}\sum_{p\in\mathcal{C}_3}}
+ \Re(\zeta_{p(010)}\zeta_{p(101)}\bar{\zeta}_{p(011)}\bar{\zeta}_{p(100)}) \Big].
\end{eqnarray}

There are 3 families of
solutions, living on the following 5-dimensional
submanifolds
\begin{eqnarray}
M_p&=&\Big\{(\zeta_k) \in \mathbb{T}^8 \,|\,
\zeta_{p(000)}\zeta_{p(111)}\bar{\zeta}_{p(001)}\bar{\zeta}_{p(110)}=+1,\nonumber\\
& &\phantom{\{(\varphi_k)  \in \mathbb{T}^n  \,|\,}
\zeta_{p(010)} \zeta_{p(101)}\bar{\zeta}_{p(100)}\bar{\zeta}_{p(011)}=+1,\nonumber\\
& &\phantom{\{(\varphi_k)  \in \mathbb{T}^n  \,|\,}
\zeta_{p(000)} \zeta_{p(011)}\bar{\zeta}_{p(001)}\bar{\zeta}_{p(010)}=-1\Big\},
\quad p\in\mathcal{C}_3.
\end{eqnarray}
Indeed, if $\zeta \in M_p$ it is an easy task to see that
\begin{eqnarray}
\begin{cases}
\zeta_{q(000)}\zeta_{q(111)}\bar{\zeta}_{q(001)}\bar{\zeta}_{q(110)}= a_{p^{-1}q(1)}\\
\zeta_{q(010)} \zeta_{q(101)}\bar{\zeta}_{q(100)}\bar{\zeta}_{q(011)}=b_{p^{-1}q(1)}\\
\zeta_{q(000)} \zeta_{q(011)}\bar{\zeta}_{q(001)}\bar{\zeta}_{q(010)}= c_{p^{-1}q(1)}\\
\zeta_{q(111)} \zeta_{q(100)}\bar{\zeta}_{q(110)}\bar{\zeta}_{q(101)}= d_{p^{-1}q(1)}
\end{cases}  \quad q\in\mathcal{C}_3,
\end{eqnarray}
where
\begin{equation}
\begin{cases}
a= (+1,-\alpha,-\alpha)\\
b= (+1,-\bar{\alpha},-\alpha)\\
c= (-1,+ \alpha,-1)\\
d= (-1,+ \alpha,-1)
\end{cases},
\label{eq:abcd}
\end{equation}
with $\alpha \in \mathbb{S}^1$ arbitrary.
Therefore, the sum in  $\pi_{\mathrm{ME}}(\zeta)$ reads
\begin{eqnarray}
& &\sum_{p\in\mathcal{C}_3} \Re\left(a+b+2 c+2 d\right)_{p^{-1}q(1)}
=\sum_{i\in S} \Re\left(a+b+2 c+2 d\right)_i \nonumber\\
& & = 2 - 8 +4\Re\alpha -4\Re\alpha=-6 ,
\end{eqnarray}
yielding $\pi_{\mathrm{ME}}(\zeta)=1/2$.

Note that, in agreement with Theorem \ref{th:pme3e},  $\pi_{\mathrm{ME}}(\zeta)$
contains $\mathcal{N}^{(4)}=12$ distinct terms that depend on phases, 6 of which are double weighted. The above solutions force 2 terms to the value + 1,  and $4\times 2$ terms to the value $=-1$. The remaining ones are symmetric around $0$ and cancel.

The corresponding uniform perfect MMES are
\begin{eqnarray}
\ket {\psi_p}&=&\frac{1}{\sqrt{8}}\Big(\zeta_{p(000)}\ket{p(000)}
+\zeta_{p(001)}\ket{p(001)}+\zeta_{p(010)}\ket{p(010)}
\nonumber\\
& &\quad\; -\bar{\zeta}_{p(000)}\zeta_{p(001)}\zeta_{p(010)}\ket{p(011)} + \zeta_{p(100)}\ket{p(100)}
\nonumber \\
& &\quad\; -\bar{\zeta}_{p(000)}\zeta_{p(001)}\zeta_{p(100)}\ket{p(101)}
+\zeta_{p(110)}\ket{p(110)}\nonumber\\
& &\quad\; +\bar{\zeta}_{p(000)}\zeta_{p(001)}\zeta_{p(110)}\ket{p(111)}\Big),
\quad p\in\mathcal{C}_3.
\label{eq:optimal3}
\end{eqnarray}
At present we do not know whether there exist other classes of uniform perfect MMES than (\ref{eq:optimal3}). Numerical evidence seems to corroborate the conjecture that (\ref{eq:optimal3}) describe all uniform perfect MMES, but we could not prove it.

\subsubsection{Real uniform MMES}
Let us now look for uniform perfect MMES whose phases are all real, i.e.\ $\zeta\in\{+1,\-1\}^8$.
A necessary condition is that $\alpha$ is real, $\alpha\in\{-1,+1\}$. In particular, it is an easy task to prove that  $\alpha=-1$  iff $a, b, c, d$ in (\ref{eq:abcd})
are permutation invariant, iff
\begin{equation}
 a=b=(1,1,1), \qquad c=d=(-1,-1,-1).
\end{equation}
Thus $\alpha=-1$ characterizes the 4-dimensional intersection
\begin{equation}
M_\star=M_{s^0}\cap M_ {s^1} =M_ {s^1}\cap M_ {s^2} = M_ {s^2}\cap M_ {s^0} =\bigcap_{p\in\mathcal{C}_3} M_p.
\end{equation}
On the other hand, $\alpha=+1$  determines the following three nonintersecting 4-dimensional submanifolds
\begin{equation}
N_p=M_p\cap \{(\zeta_k) \in \mathbb{T}^8 | \alpha=\zeta_{p(000)}\zeta_{p(110)}\bar{\zeta}_{p(010)}\bar{\zeta}_{p(100)}=+1\}, \qquad p\in \mathcal{C}_3.
\end{equation}
Therefore,  all uniform perfect MMES with real $\zeta$ belongs to one of the above nonintersecting manifolds, namely
\begin{equation}
\{\text{real uniform MMES}\}  \subset \bigcup_{p\in\mathcal{C}_3} N_p \cup M_\star.
\end{equation}
Thus the total number of real uniform perfect MMES is $4 \times 2^4=2^6$.
They are given by
\begin{equation}
\ket{\psi}=\frac{1}{\sqrt{8}} \sum_{k\in X^3} \zeta_k \ket{k}, \qquad
\zeta\in \{-1,1\}^8
\end{equation}
with
\begin{equation}
\begin{cases}
\zeta_{000}\zeta_{001}\zeta_{010}\zeta_{011}= x_j\\
\zeta_{000} \zeta_{001}\zeta_{100}\zeta_{101}=y_j\\
\zeta_{000} \zeta_{010}\zeta_{100}\zeta_{110}= z_j\\
\zeta_{001} \zeta_{010}\zeta_{100}\zeta_{111}= w_j,
\end{cases}
\qquad 1\leq j\leq 4,
\end{equation}
where
\begin{eqnarray}
& & x= (-1,-1,-1,+1), \quad
y = (-1,-1,+1,-1),\nonumber\\
& &z= (-1,+1,-1,-1),\quad
w= (-1,+1,+1,+1).
\end{eqnarray}

\subsection{$n>3$ qubits}

For $n=4$ qubits,  $N=16$, $N_A=N_{\bar{A}}=4$, and a brute force enumeration shows that the minimum value of the potential of multipartite entanglement in  the class of real uniform states,  is
\begin{equation}
\min \left\{\pi_{\mathrm{ME}}^{(4)}(\zeta) \,|\, \zeta\in \{-1,1\}^{16} \right\}= \frac{1}{3}>\frac{1}{4}.
\end{equation}
In fact, there are $1056$ minimizers, among which, there is, e.g.
\begin{equation}
\zeta=(-1,-1,-1,-1,-1,-1,+1,+1, -1,+1,-1,+1,+1,-1,-1,+1).
\label{eq:MMES4}
\end{equation}
There is numerical evidence that $1/3$ is the minimum of the multipartite entanglement, and thus it is not an artifact of the restriction to real uniform states. In fact, it has been proved that for $n=4$ the minimum of $\pi_{\mathrm{ME}}$ is strictly larger than $1/4$ \cite{sudbery1,sudbery}, but still its value is  unknown  \cite{borras}.
This is a first example of
frustration among the bipartitions, that prevents the existence of a
perfect MMES: the requirement that purity be
minimal for all balanced bipartitions generate conflicts already for
$n=4$ qubits.

For $n=5$ and 6, the expressions become more complicate. Here, we will not discuss this cases. We will only exhibit two real uniform perfect MMES,  solutions to
\begin{equation}
\pi_{\mathrm{ME}}^{(5)}(\zeta)=\frac{1}{4}, \qquad \zeta\in\{-1,+1\}^{32}
\end{equation}
and
\begin{equation}
\pi_{\mathrm{ME}}^{(6)}(\zeta)=\frac{1}{8}, \qquad \zeta\in\{-1,+1\}^{64},
\end{equation}
respectively. Therefore,  interestingly, frustration is present for
$n=4$ qubits, while it is absent for $n=5$ and 6.

For example, a
5-qubits real uniform perfect MMES is defined by Eq.\ (\ref{eq:randomphase}) with
the following set of phases
\begin{eqnarray}
\zeta &=&(+1,+1,+1,+1,+1,-1,-1,+1, +1,-1,-1,+1,+1,+1,+1,+1,\nonumber\\
& &\;\;\!+1,+1,-1,-1,+1,-1,+1,-1,-1,+1,-1,+1,-1,-1,+1,+1) \nonumber\\
\label{eq:ideal5}
\end{eqnarray}
and can be shown to live on a 7-dimensional manifold,
while a
6-qubits real uniform perfect MMES has
the following set of phases
\begin{eqnarray}
\zeta &=&(+1,+1,-1,+1,-1,-1,-1,+1, -1,-1,+1,-1,-1,-1,-1,+1,\nonumber\\
& &\;\;\!-1,+1,-1,-1,-1,+1,+1,+1,-1,+1,-1,-1,+1,-1,-1,-1,\nonumber\\
& &\;\;\!+1,-1,-1,-1,-1,+1,-1,-1,+1,-1,-1,-1,+1,-1,+1,+1,\nonumber\\
& &\;\;\!+1,+1,+1,-1,+1,+1,-1,+1,-1,-1,-1,+1,+1,+1,-1,+1). \nonumber\\
\label{eq:MMES6}
\end{eqnarray}

By using the theory of  quantum weight enumerators and quantum codes \cite{rains1,rains3,rains2}, it has been proved that \cite{scott}
\begin{equation}
\min\{ \pi_{\mathrm{ME}}^{(n)}(z)\,|\, z\in \mathbb{S}^{2^{n+1}-1} \}> 2^{-[n/2]}, \quad \text{for}\; n\geq 8,
\end{equation}
and thus there is frustration among the bipartitions that prevents the existence of a
$n$-qubit perfect MMES, for $n\geq 8$.
The case $n=7$ is still open. There is numerical evidence that it is frustrated too, but no conclusive arguments.

Summarizing, perfect MMES exist for $n=2, 3, 5, 6$ and, possibly, for $n=7$. For $n=4$ and $n\geq 8$ there is frustration and the minimum of the potential of multipartite entanglement is strictly larger than $2^{-[n/2]}$.
Interestingly enough, in the cases considered ($n\leq6$) we have shown that the (conjectured) minimum of the potential is attained by uniform states with real phases. In such a case, in order to study the structure of multipartite entanglement in a quantum state of $n$ qubits, and in particular the minima of its potential,  one can  instead consider the simpler system of  classical sequences $\zeta\in\{-1,+1\}^{2^n}$ of $2^n$ bits, with Hamiltonian $\pi_{\mathrm{ME}}^{(n)}(\zeta)$.

\section{Conclusions}
\label{sec:conclusions}

In this paper we have studied  the properties  of the potential of multipartite entanglement  and of its minimizers, the MMES, for a system of $n$ qubits. In particular our focus has been on perfect MMES, that saturate the lower bound of the potential,
and by using a probabilistic approach, we have proven a theorem on the structure of their population probability vectors.  This allowed us to consider a particular simple class of solutions, those with uniform population. We have shown by explicit construction that (apart for the case $n=7$ which is still open, but probably is frustrated) there always exist uniform perfect MMES with real phases, a class of states  that can be mapped to the classical binary sequences of length $2^n$. In fact, we have shown that also for $n=4$, the lowest number at which frustration occurs and hinders the existence of perfect MMES, the (conjectured) minimum of the potential of multipartite entanglement is attained   by uniform states with real phases.
This represents a great advantage, because  in this situation one can investigate the structure of quantum multipartite entanglement by studying the simplest problem a classical Hamiltonian defined on  binary sequences.

\section*{Acknowledgements}
I would like to thank G. Florio, U. Marzolino, G. Parisi, and S. Pascazio for
many conversations and stimulating discussions on multipartite entanglement.
One of the reasons for having written this article is the enthusiasm of S. Graffi
for the subject; I would like to thank him for this.
This work is partly supported by the European Community through the Integrated Project EuroSQIP.

\end{document}